\documentclass{elsarticle}

\usepackage{amsmath}
\usepackage{amsthm}
\usepackage{amssymb}
\usepackage{graphicx}
\usepackage{verbatim}
\usepackage{natbib}
\usepackage{subfigure}
\usepackage[margin=1.35in]{geometry}
\usepackage{hyperref}

\theoremstyle{plain}
\newtheorem{prop}{Proposition}[section]

\newtheorem{lemm}[prop]{Lemma}
\newtheorem{theo}[prop]{Theorem}

\theoremstyle{definition}

\newtheorem{defi}[prop]{Definition}

\theoremstyle{remark}

\newcommand{\gh}{\hat{\gamma}}

\newcommand{\eqd}{\,{\buildrel d \over =}\,}
\newcommand{\hk}{\hat{k}}

\newcommand{\EE}{\mathbb{E}}
\newcommand{\NN}{\mathbb{N}}
\newcommand{\PP}{\mathbb{P}}
\newcommand{\RR}{\mathbb{R}}

\newcommand{\nn}{\mathcal{N}}
\newcommand{\cc}{\mathcal{C}}
\newcommand{\law}{\mathcal{L}}

\newcommand{\Var}{\operatorname{Var}}
\newcommand{\Cov}{\operatorname{Cov}}

\newcommand{\argmin}{\operatorname{argmin}}
\newcommand{\sgn}{\operatorname{sgn}}

\newcommand{\htau}{\hat{\tau}}
\newcommand{\tRpt}{Y_\tau}

\newcommand{\limc}{\lim_{c \rightarrow \infty}}

\newcommand{\limn}{\lim_{n \rightarrow \infty}}
\newcommand{\lims}{\lim_{s \rightarrow \infty}}
\newcommand{\limt}{\lim_{t \rightarrow \infty}}

\newcommand{\ABS}{We study a new estimator for the tail index of a distribution in the Fr\'echet domain of attraction that arises naturally by computing subsample maxima. This estimator is equivalent to taking a $U$-statistic over a Hill estimator with two order statistics. The estimator presents multiple advantages over the Hill estimator. In particular, it has asymptotically $\cc^\infty$ sample paths as a function of the threshold $k$, making it considerably more stable than the Hill estimator. The estimator also admits a simple and intuitive threshold selection rule that does not require fitting a second-order model.}

\begin{document}

\title{Subsampling Extremes:\\ From Block Maxima to Smooth Tail Estimation}

\author{Stefan Wager \\ Department of Statistics \\ Stanford University \\ \texttt{{swager@stanford.edu}}}
\date{\today}

\begin{abstract}

\ABS

\end{abstract}

\maketitle

\section{Introduction}
\label{sec:intro}

Researchers in multiple fields face a growing need to understand the tails of probability distributions, and extreme value theory presents tools which, under certain regularity assumptions, let us build simple yet powerful models for these tails. In the case of heavy tailed distributions, the setting of extreme value theory is as follows: suppose our data is drawn from a distribution $F$, and assume that there is a constant $\gamma > 0$ and some slowly varying function $L$ such that
\begin{equation}
1 - F(x) = L(x) \cdot x^{-\frac{1}{\gamma}}, \ \text{with} \lim_{x \rightarrow \infty} \frac{L(ax)}{L(x)} = 1 \ \text{for all} \ a > 0.
\label{eq:frechet}
\end{equation}
Then, $F$ is in what is called the Fr\'echet domain of attraction.
If $F$ satisfies this property (which most commonly used heavy-tailed distributions do), extreme value theory provides an elegant and concise description of the asymptotic properties of sample maxima of $F$. A major challenge is that this description relies on knowledge of the parameter $\gamma$, called the tail index of the distribution $F$. And, unfortunately, estimating $\gamma$ from data is not always straightforward.

The literature on tail index estimation is quite extensive. One of the most widely used estimators is due to \citet{hill1975simple}, who suggests estimating $\gamma$ with a simple functional of the top $k + 1$ order statistics of the empirical distribution:
\begin{equation}
\label{eq:hill}
\gh_H := \frac{1}{k}\sum_{j = 0}^{k-1} \log \left [\frac{X_{n - j, n}}{X_{n - k, n}}\right].
\end{equation}
Here, $X_{1, \, n} \leq \ldots \leq X_{n, \, n}$ denote the order statistics of $X$, and $k$ must be selected such that $X_{n - k, \, n} > 0$.
Hill showed that $\gh_H$ converges in probability to $\gamma > 0$, provided the threshold sequence $k = k(n)$ is an intermediate sequence that grows to infinity slower than the sample size $n$. Hill's idea of using a functional of extreme and intermediate order statistics to estimate $\gamma$ has received considerable attention. \citet{csorgo1985kernel} suggest ways to adaptively weight the order statistics, while \citet{dekkers1989moment} modify Hill's estimator so that it is also consistent for a generalization of \eqref{eq:frechet} that includes negative $\gamma$. There have been proposals to eliminate the asymptotic bias of the Hill estimator \citep{beirlant1999tail,feuerverger1999estimating,gomes2000alternatives,peng1998asymptotically}; recent proposals \citep{caeiro2005direct,gomes2007sturdy,gomes2008tail} show how to do so without increasing asymptotic variance.

Nonetheless, tail index estimation remains quite challenging, especially for smaller samples on the order of a few hundred to a thousand points. Of course, many difficulties are inherent to the subject matter: only a small fraction of any sample will be inside the tail of the underlying distribution, and so even large samples may contain very little information relevant to inference about this tail.

Other challenges, however, seem to arise from specifics of popular estimators. All estimators for $\gamma$ require choosing a threshold at which the tail area of the distribution begins. Ideally, specifying a good threshold should be easy, and the estimate $\gh$ should not be sensitive to small changes in the threshold. Unfortunately, most commonly used estimators for $\gamma$ do not reach this ideal. In the case of the Hill estimator---where the parameter $k$ from \eqref{eq:hill} stands in for the threshold---the choice is far from innocuous:

\begin{itemize}
\item Inadequate choice of $k$ can lead to large expected error. Small values of $k$ lead to high variance, while large values of $k$ usually lead to high bias. There is often an intermediate region for $k$ where the estimator has fairly small expected error, but it is not always easy to find this region.
\item The Hill estimator is extremely sensitive to small changes in $k$, even asymptotically: \citet{mason1994weak} show that the Hill estimator process converges in law to a modified Brownian motion. Thus, even within the `good' region with low expected error, a minute change in $k$ can impact the conclusions to be drawn from the model.
\end{itemize}

The problem of choosing the threshold $k$ has been discussed, among others, by \citet{beirlant2002exponential}, \citet{danielsson2001using}, \citet{drees1998selecting}, and \citet{guillou2001diagnostic}. Most existing methods rely on fairly complicated auxiliary models: all but the last of the cited ones require either implicitly or explicitly fitting a difficult-to-fit second-order convergence parameter. As the method due to Guillou and Hall does not require fitting secondary parameters, we use it as our main benchmark in simulation studies. The problem of excessive oscillation of the Hill estimator has been discussed by \citet{resnick1997smoothing}, who recommend smoothing the Hill estimator by integrating it over a moving window. We are not aware of any guidance on how to automatically select $k$ for this smoothed Hill estimator.

In this paper, we study a new estimator for $\gamma$ that arises from a simple subsampling idea. It is well known that sample maxima from a distribution $F$ satisfying \eqref{eq:frechet} have the following property: if $X_1, ..., X_n$ are drawn independently from $F$, then
$$ \limn \PP\left[\frac{\max\{X_1, ..., X_n\}}{\ell(n) \cdot n^\gamma}  \leq x\right] =  G_\gamma(x), $$
where $G_\gamma$ is a limiting cumulative distribution function that only depends on $\gamma$ and $\ell(n)$ is an appropriately chosen slowly varying function. Noting this, we may suspect that when $F$ has positive support,
\begin{equation}
\lim_{s \rightarrow \infty} s \cdot (\EE[\log \max\{X_1, ..., X_s\}] - \EE[\log\max\{X_1, ..., X_{s - 1}\}]) = \gamma.
\label{eq:first}
\end{equation}
In Theorem \ref{theo:rbm}, we show that this relation in fact holds under very mild conditions on $F$ near 0. Our estimator follows directly from this formula. Given a subsample size $1 < s \leq n$, we first estimate the quantities
$$\EE[\log \max\{X_1, ..., X_s\}] \text{ and } \EE[\log \max\{X_1, ..., X_{s - 1}\}]$$
by subsampling our data without replacement, and then use \eqref{eq:first} to obtain an estimate for $\gh$. Since this estimator operates by computing the average log maxima of randomly subsampled blocks, we call it the Random Block Maxima (RBM) estimator. Our idea is related to proposals for tail index estimation that study weighted sums of log-ratios of order statistics by, e.g., \citet{drees1998smooth} and \citet{gardes2008moving}.

The growth rate of subsample functionals for heavy-tailed data has also been studied by \citet{bertail2004subsampling}, \citet{mcelroy2007computer}, and \citet{politis1999subsampling}, who show how to estimate the tail index $\gamma$ using Monte Carlo analysis on subsamples. The point of view taken in their papers is closely related to bagging \citep{breiman1996bagging,buhlmann2002analyzing,buja2006observations}. Here, however, we only use subsampling to motivate our procedure; the final estimator $\gh$ defined in \eqref{eq:gh} can efficiently be computed in closed form.

The RBM estimator can be understood as belonging to two different frameworks of tail index estimation. The block maxima approach, which was often used in the early days of extreme value theory, aims to directly fit the distribution of the maxima of fixed (e.g., yearly) blocks of data.
See \citet{gumbel1958statistics} for a review; \citet{dombry2013maximum} and \citet{ferreira2013block} provide a modern analysis.
In this light, the RBM estimator can be seen as a randomized method of moments estimator in the block maxima framework. Our estimator, however, can also be seen as an outgrowth of the more modern tail estimation paradigm started by the Hill estimator: as we will show, the RBM estimator can be constructed by taking a $U$-statistic over a Hill estimator with two order statistics. In other words, once we start subsampling the data, the block maxima and Hill estimation frameworks merge and lead to the RBM estimator.

Our estimator behaves much like the Hill estimator; however, it addresses threshold selection much more naturally than the latter:
\begin{itemize}
\item The RBM estimator has asymptotically smooth sample paths as a function of its threshold parameter $k$ as defined in \eqref{eq:k}, and, even in modestly sized samples, does not suffer from small-scale instability in $k$.
\item Thanks to its smoothness properties, the RBM estimator admits a simple and intuitive threshold selection rule that does not require fitting a second-order model.
\end{itemize}

\begin{figure}[t]
\centering
\includegraphics[width=0.55\textwidth]{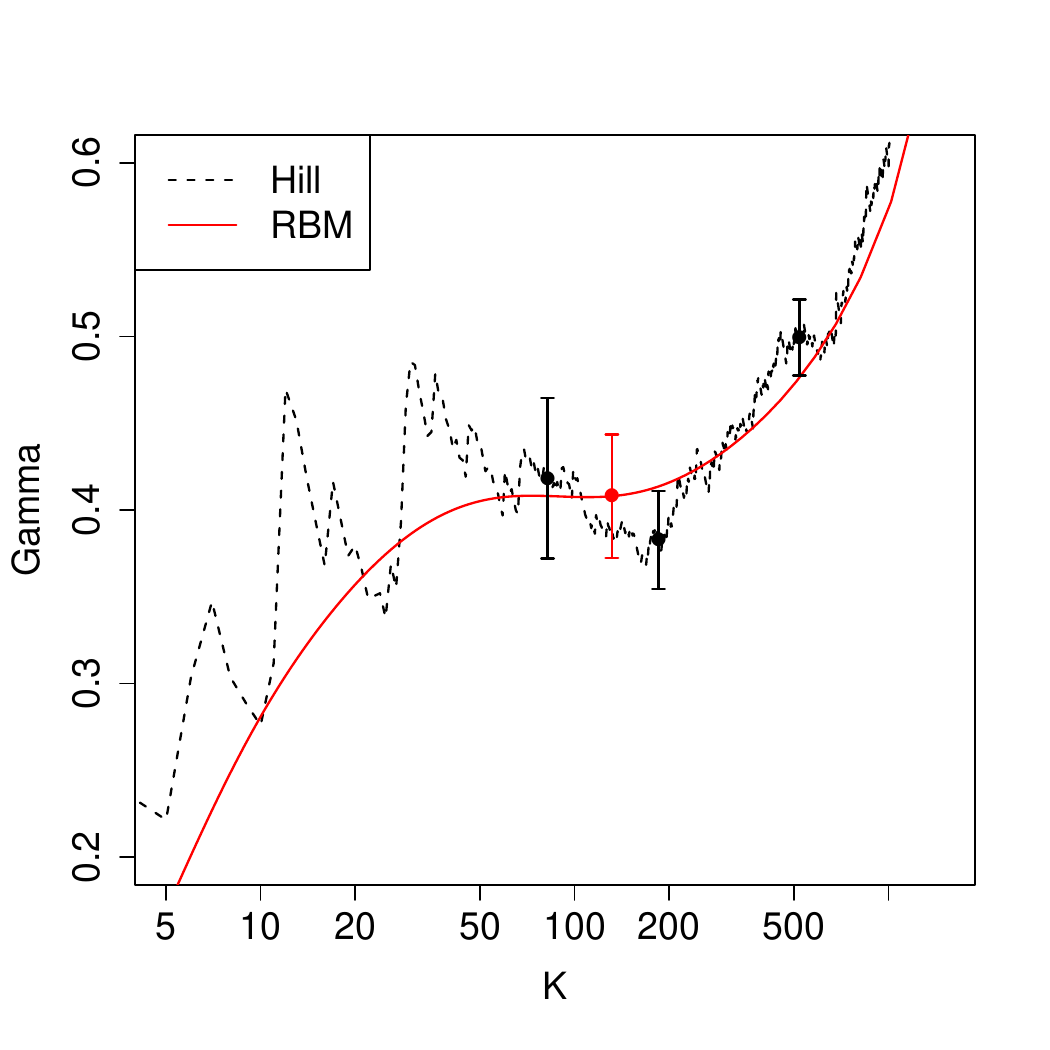}
\caption{Estimates $\gh$ produced by both the Hill and RBM estimators on an IPO dataset ($n = 2037$), as a function of a threshold parameter $k$. We have highlighted some plausible choices for $(\hk, \, \gh)$, along with error bars that are 1 standard error long in each direction (the error bars assume a pre-determined choice of $k$). This dataset is discussed in detail by \citet{petty2012economies}.}
\label{fig:warm_up}
\end{figure}

Figure \ref{fig:warm_up} shows estimates produced by both the Hill and RBM estimators for the tail index $\gamma$ of gross proceeds from venture capital backed IPOs in the United States between 1995 and 2011. Both estimates depend on a threshold parameter $k$. As expected, the RBM sample path is much smoother than the Hill sample path. This makes it easier to select $k$ with the RBM estimator than with the Hill estimator. According to standard guidelines, we should select the tuning parameter $\hk$ for the Hill estimator from an area where the estimator $\gh(k)$ is as stable as possible. But, in Figure \ref{fig:warm_up}, we highlight three different plausible-looking choices for $\hk$ that lead to substantially different estimates $\gh$. We can try to avert this problem by selecting the threshold $\hk$ using a formal rule, but this is not necessarily a fool-proof strategy. It may be possible, for example, for an overzealous practitioner to tune the formal rule to get the plausible-looking answer of his choice. With the RBM estimator, conversely, there is no question about what the right threshold should be, and there is no room for confusion or for second-guessing a choice of $\hk$. In Figure \ref{fig:warm_up}, the threshold $\hk$ for the RBM estimator was selected using our automatic rule, but a human analyst would no doubt have chosen the same threshold.

As our examples and simulation study  should make clear, the main advantage of the RBM estimator is not that it beats the state-of-the-art in tail index estimation by having low mean squared error (MSE). Rather, its strength lies in its stability and ease of use. Practitioners using the RBM estimator can get close to optimal estimates for $\gamma$ by using an estimator $\gh(k)$ that is smooth in the tuning parameter $k$. We have already emphasized that this smoothness facilitates threshold selection, but the advantages do not stop there:
\begin{itemize}
\item The RBM estimator is stable enough in $k$ that we can visually inspect the quality of the extreme value theoretic model and look for abnormal patterns that may indicate a failure of modeling assumptions by simply examining a plot of $\gh$ against $k$. In comparison, the corresponding curve for the Hill estimator is so noisy that it can be difficult to pick out any meaningful patterns with the naked eye.
\item The smooth relationship between $\gh$ and $k$ allows us to use labeled training data to choose $k$ by supervised risk minimization -- e.g., by running RBM on multiple datasets of the same size as our dataset of interest and with known $\gamma$, and then picking $k$ with the lowest prediction error. With the Hill estimator, the noise level is high enough that selection bias can easily overwhelm any true signal; however, with the RBM estimator the number of local minima to choose from is small and so the risk of problems related to selection bias is greatly reduced.
\item With the RBM estimator, a small change in $k$ will usually not produce a large change in $\gh$, and so it is more difficult for a marginally honest experimentalist to tune his choice of $k$ in such a way as to get the value of $\gh$ he wants. Thus, in controversial situations, the RBM estimator may allow for less experimental bias than the Hill estimator.
\end{itemize}
Finally when paired with our threshold selection method, the RBM estimator allows us to get a point estimate for $\gamma$ without having to fit a second-order model and without having to resort to manual threshold selection (for example, \citet{coles2001introduction} recommends manually examining a ``mean residual life plot" to select a threshold when estimating $\gamma$ by maximum likelihood).
In other words, without compromising quality, our RBM estimator is easier to use and gives more stable estimates for $\gamma$ than the Hill estimator, which is one of the most widely used tools for estimating the tail index of a heavy-tailed distribution.

\section{Random Block Maxima}
\label{sec:howto}

As described in \eqref{eq:first}, the RBM estimator for a given subsample size $s \in \{2, \, ..., \, n\}$ is defined by
\begin{equation}
\gh_{RBM}(s) = s \cdot (M(s) - M(s - 1)),
\label{eq:gh}
\end{equation}
where $M(s)$ is the average log maximum over all subsamples of size $s$ drawn without replacement from the full sample of size $n$
\begin{equation}
M(s) = \binom{n}{s}^{-1} \sum_{1 \leq i_1 < ... < i_s \leq n} \max_{\{1 \leq j \leq s\}} \{\log X_{i_j}\},
\label{eq:max}
\end{equation}
and $X_{i_j}$ denotes the $i_j$-th element of $\{X_1, \, ..., \, X_n\}$.
Since we are interested in the behavior of sample maxima, we need to use resampling without replacement instead of with replacement. Otherwise, the presence of duplicate elements in our subsamples would bias our estimates $M(s)$ downwards. In practice we use the formula
\begin{equation}
\label{eq:rbm_compute}
M(s) = \sum_{j = 1}^{n-s+1}  \binom{n - j}{s - 1} \Big/ \binom{n}{s} \, \log X_{n - j + 1, n},
\end{equation}
which allows for efficient computation.

To facilitate comparison between the Hill and RBM estimators, we do not parametrize our estimator directly in terms of the subsample size $s$, but use
\begin{equation}
\label{eq:k}
k = \frac{2n}{s}.
\end{equation}
The motivation for this transformation is that it allows us to match the Hill and RBM estimators by their variance. The Hill and RBM estimators then both have high variance for small $k$ and potentially high bias for large $k$. More precisely, as shown in Theorem \ref{theo:rbm}, the RBM estimator has asymptotic variance
$$ \lim_{n\rightarrow\infty} k(n)\widetilde{\Var}[\gh_{RBM}(k(n))] = \gamma^2, $$
for any intermediate sequence $k(n) = 2n/s(n)$, just like the Hill estimator (the asymptotic variance is the variance of the limiting normal distribution). Asymptotic bias increases with $k$ at a rate that depends on second-order parameters. When there is no risk of confusion, we will sometimes write $\gh_{RBM}(k)$ instead of $\gh_{RBM}(s)$, where $k$ and $s$ are understood to be connected by \eqref{eq:k}.

It is useful to plot $\gh_{RBM}(k)$ against $k$, which gives us an analog of a Hill plot. We have found such plots to be most informative when we plot $k$ on a log scale rather than on a linear scale, as recommended by \citet{drees2000make}. Once we have computed $\gh_{RBM}(k)$ at multiple $k$, the problem becomes to choose which threshold $\hat{k}$ to use for estimating $\gamma$. A good choice of threshold $k$ should aim to simultaneously keep the bias and variance components small.

As we show in section \ref{sec:process}, our estimator $\gh_{RBM}$ converges weakly to a $\cc^{\infty}$ limiting process. In practice, $\gh_{RBM}$ is smooth enough as a function of $k$ that we can reliably estimate its derivative in finite samples. This enables a particularly simple method for selecting a threshold $\hat{k}$ at which to report $\gh$.

We start by computing $\gh_{RBM}$ for subsample sizes $s = n, \, n-1, \, ..., \, 2$. By \eqref{eq:k}, these choices of $s$ correspond to $k$-values $k_1 < k_2 < ... < k_{n-1}$, where $k_i = 2n/s_i$. We then pick $k$ using
\begin{equation}
\hat{k} = \argmin_{s} \left\{\left(\frac{\gh_{RBM}(k_s) - \gh_{RBM}(k_{s-1})}{\log k_s - \log k_{s-1}}\right)^2 + \frac{\gh^2_{RBM}(k_s)}{2 k_s}\right\}.
\label{eq:rbmopt}
\end{equation}
Roughly speaking, this heuristic aims to minimize the square of the derivative
$$ \frac{\partial}{\partial\log k} \ \gh_{RBM}(k) $$
subject to a penalty term that decays as $1/k$. As argued in section \ref{sec:opt}, our choice of $\hat{k}$ aims to minimize possible bias in a heuristic Bayesian sense. We note that this threshold selection procedure is dependent on the smoothness properties of $\gh_{RBM}$. Attempting to use the same method with the Hill estimator $\gh_H$ would not lead to good results, since $\gh_H$ is not asymptotically differentiable as a function of $k$.

\subsection{Examples}
\label{sec:case}

\begin{figure}[t]
\centering
\subfigure[$N = 2000$ points drawn from a Student-$t$ distribution with 4 df. ($\gamma = 0.25$)]{
\includegraphics[width=0.45\textwidth]{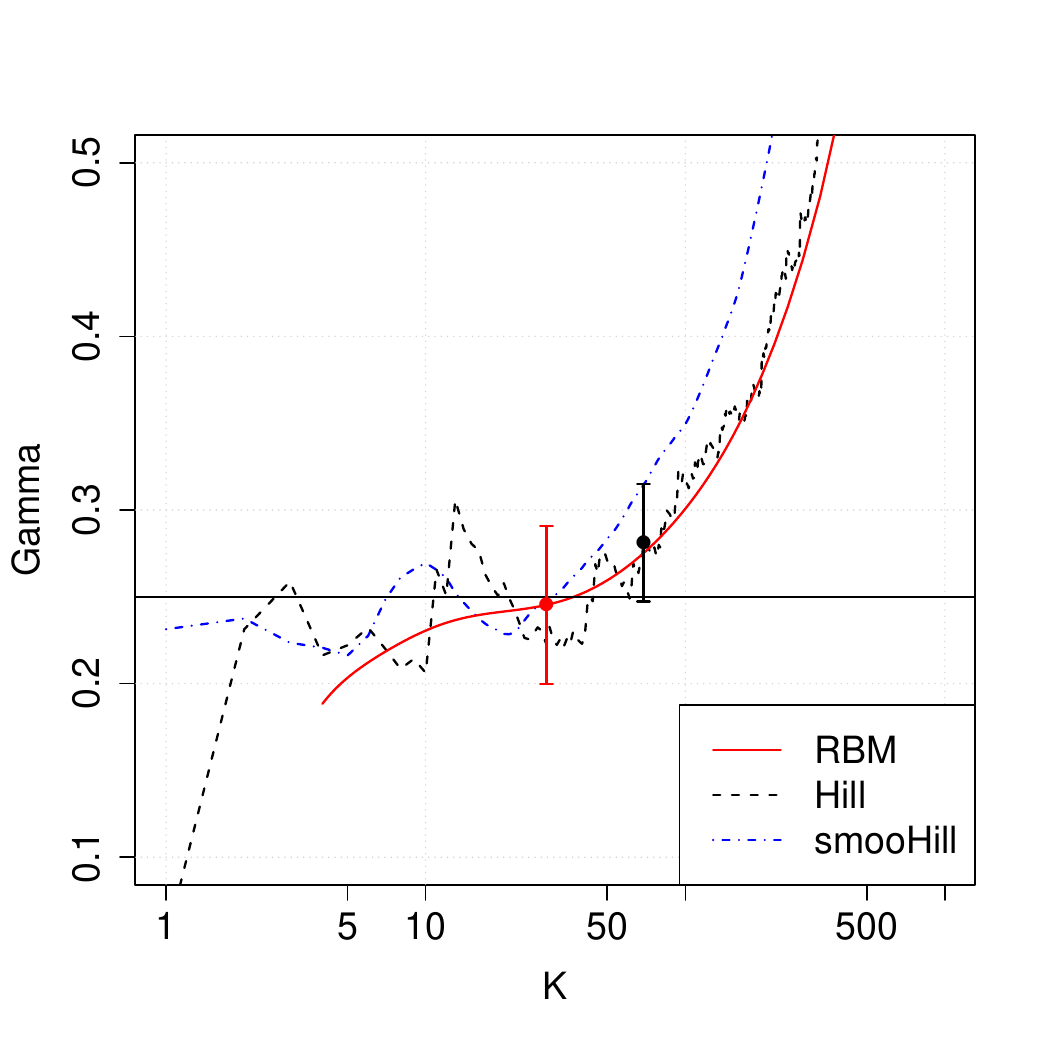}
\label{fig:st4}
}
\subfigure[$N = 500$ points drawn from a Fr\'echet distribution with shape parameter $\chi = 2$. ($\gamma = 0.5$)]{
\includegraphics[width=0.45\textwidth]{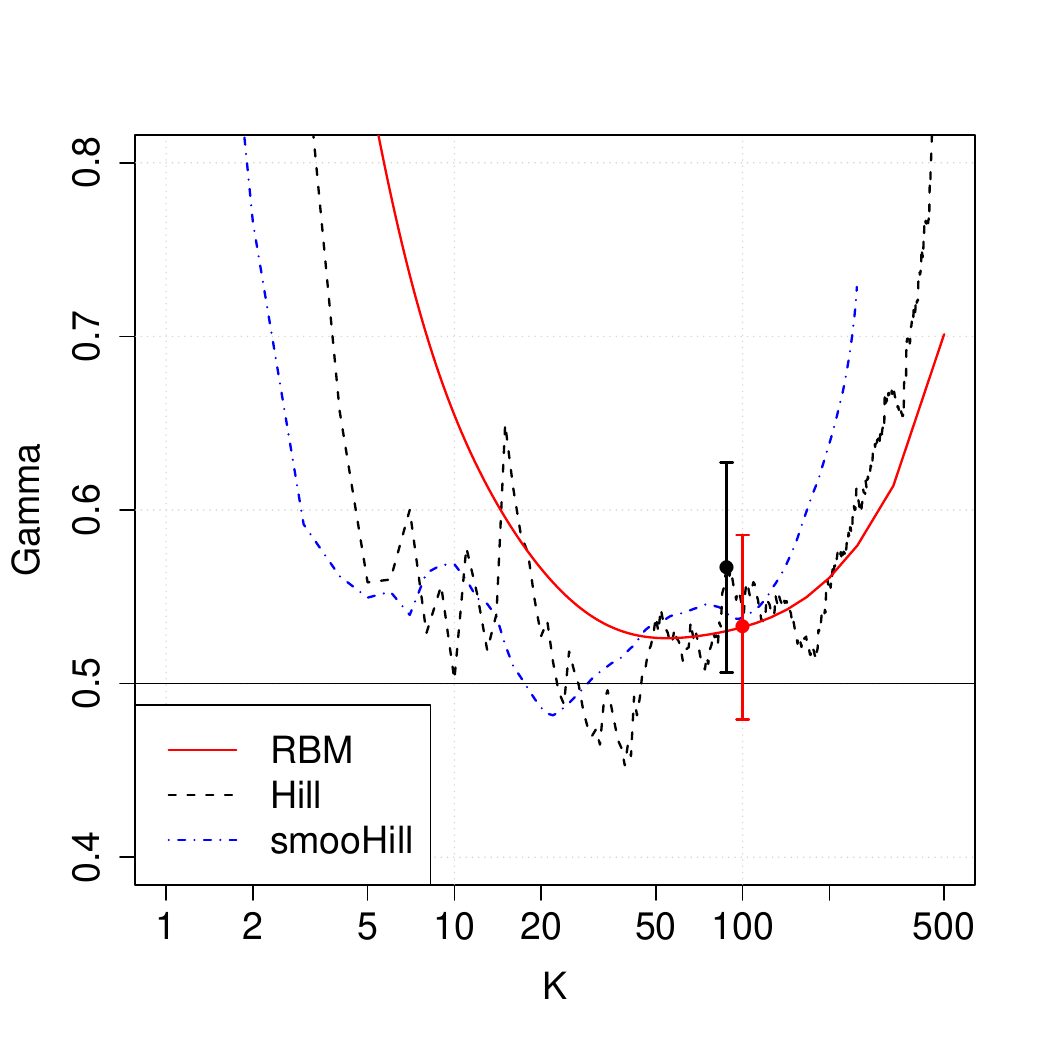}
\label{fig:f2}
}
\caption{Comparison of the RBM, Hill, and smooHill estimators on simulated data. The true value of $\gamma$ is shown by a horizontal line, and the error bars are 1 standard deviation wide. The error bars assume a pre-determined choice of $k$, and do not account for tuning.}
\label{fig:examples}
\end{figure}

We compare the RBM estimator to both the Hill estimator and the smoothed Hill estimator (smooHill) proposed by \citet{resnick1997smoothing}. There exist various heuristics for how wide a smoothing window to use for the smooHill estimator. We follow \citet{resnick2007heavy} and average the Hill estimator on $(k, 2k]$ for each $k$. For the Hill estimator, we use the method from \citet{guillou2001diagnostic} to automatically select $k$, while for RBM we use $\hat{k}$ from \eqref{eq:rbmopt}. The goal of these examples is to show how the RBM estimator can be used in practice; a more rigorous simulation study is given in section \ref{sec:simulation}.

We applied all three estimators first to 2000 datapoints drawn independently from a Student-$t$ distribution with 4 degrees of freedom ($\gamma = 0.25$), and then to 500 datapoints from a Fr\'echet distribution with a shape parameter of 2 ($\gamma = 0.5$). In the case of the Student-$t$ distribution, we discarded all negative datapoints (since all considered estimators involve taking logs of the datapoints), giving us an effective sample size of 992. Our results are given in Figure \ref{fig:examples}.

We observe that the RBM estimator oscillates much less than the Hill estimator or even the smooHill estimator (which has asymptotically $\cc^1$ sample paths whereas the RBM estimator is asymptotically smooth). The instability of the Hill estimator is not benign: around the selected threshold, a small change in $k$ can shift the confidence interval for the estimator by a full standard deviation and potentially change conclusions drawn from the model. Thus, although the estimates given by the RBM estimator at the selected thresholds are not more accurate than those given by either the Hill or the smooHill at the same thresholds, they are much less ambiguous. This should be quite useful in applications, since the less ambiguous the answers given by an estimator are, the easier it is to evaluate convergence, and the less room there is for data dredging or other types of confusion.

\subsection{Avoiding pitfalls}
\label{sec:warnings}

Finally we highlight a few cases where the RBM estimator as described here can fail, and show how to avoid these cases. First, the RBM estimator is somewhat computationally intensive. Our implementation based on \eqref{eq:rbm_compute} can comfortably handle cases where $n$ ranges in the low thousands; however, it becomes painfully slow when $n$ approaches hundreds of thousands.
One way to avoid this problem without losing much information is to throw out all but the largest $M$ datapoints (we usually take $M = 2'000$ or $10'000$). This speeds up the algorithm a lot, and does not cost much in terms of accuracy because most of the information relevant to estimating $\gamma$ is in the largest datapoints anyways.

Second, our threshold selection heuristic may fail if the data given to the RBM estimator is predominantly not from the tail of the distribution; this issue is discussed further in section \ref{sec:opt}. Again, a solution to this problem is to filter our data; in this case, we may want to throw out all the data that does not appear to be in the tail area we are trying to model.

\section{Asymptotics of Random Block Maxima}
\label{sec:rbm}

We now move to theoretical results. The limiting distribution of the RBM estimator can largely be derived from the theory of $U$-statistics. A $U$-statistic is a multi-parameter generalization of a sample mean: given data $X_1, ..., X_n$ and a symmetric $s$-parameter function $f$, the $U$-statistic over $f$ is defined as
\begin{equation}
\label{eq:u}
U_n\left(X_1, ..., X_n\right) := \binom{n}{s}^{-1} \sum_{\{I \subseteq \{1, ..., n\}:|I| = s\}} f\left(X_I\right),
\end{equation}
where $|\cdot|$ denotes cardinality.
Such statistics have many desirable regularity properties. In particular, \citet{hoeffding1948class} showed that when the underlying function $f$ is held fixed, $U$-statistics are asymptotically normal with variance decaying as $1/n$.

As we have already stated earlier, our estimator $RBM_{k, n} := \gh_{RBM}(k)$ given in \eqref{eq:gh} can be described as a $U$-statistic over the Hill estimator. More precisely, for positive random variables $X_1, ..., X_s$ with $s^{th}$ order statistics $X_{1, s} \leq ... \leq X_{s, s}$, let $H_1^{(s)}$ be the $k = 1$ Hill estimator on $s$ datapoints
\begin{equation}
\label{eq:h1}
H_1^{(s)}(X_1, ..., X_s) := \log X_{s, s} - \log X_{s-1, s}.
\end{equation}
We can then write $RBM_{k, n}$ as a $U$-statistic over $H_1^{(s)}$. All proofs are given in the Appendix.

\begin{lemm}
\label{lemm:notation}
Let $X_1, ..., X_n$ be positive random variables with $n^{th}$ order statistics $X_{1, n} \leq ... \leq X_{n, n}$. Then, the RBM estimator given in \eqref{eq:gh} is equal to
$$ RBM_{k, n} = \binom{n}{s}^{-1} \sum_{1 \leq i_1 < ... < i_s \leq n} H_1^{(s)}(X_{i_1}, ..., X_{i_s}), $$
where $k$ satisfies the relation $s =  2n/k $.
\end{lemm}

Expressing $RBM_{k, n}$ as a $U$-statistic enables us to leverage the extensive literature on the topic. Our problem, however, does not quite fall into the classical scope of $U$-statistics. Most of the literature assumes that the function $f$ in \eqref{eq:u} is fixed as $n$ grows. But, in our case, the functions $H_1^{(s)}$ take a number of parameters that increases with $n$. Such a $U$-statistic is called an infinite order $U$-statistic. Although (as shown below) the classical asymptotic distributional results for $U$-statistics still hold in our case, the infinite order nature of the problem requires some additional work.

A common strategy for showing the asymptotic normality of a sequence of statistics $\{U_n\}$ is by approximating the $U_n$ by their H{\'a}jek projections $\widehat{U}_n$. Suppose $X_1, ..., X_n$ are drawn from some known distribution, and let $U_n$ be an $n$-parameter function. We then define its H{\'a}jek projection $\widehat{U}_n$ as
\begin{equation}
\label{eq:hajek}
\widehat{U}_n := \EE\left[U_n\right] + \sum_{i = 1}^n \EE\left[U_n - \EE\left[U_n\right]|X_i\right].
\end{equation}
The advantage of studying such projections is that, when the $X_i$ are independent and identically distributed (\emph{iid}), $\widehat{U}_n$ is a sum of independent random variables to which we can apply the central limit theorem.

When $U_n$ is a $U$-statistic, the difference $U_n - \widehat{U}_n$ converges to zero in mean square under fairly general conditions. The following lemma is a consequence of the Efron-Stein ANOVA decomposition.

\begin{lemm}
\label{lemm:anova}
Let $X_1, X_2, ...$ be \emph{iid} random variables, and let $s(n)$ be a sequence such that $s(n) \leq n$ for all $n$. Moreover, let $g^{(n)}$ be a sequence of real-valued $s(n)$-parameter functions that are symmetric in their arguments, and let there be a constant $C$ such that
$$\Var\left[g^{(n)}(X_1, ..., X_{s(n)})\right] \leq C$$
for all $n$. Then, taking $U_n$ as a $U$-statistic over $g^{(n)}$
$$ U_n = \binom{n}{s(n)}^{-1} \sum_{\{I \subseteq \{1, ..., n\}: |I| = s(n)\}} g^{(n)}(X_I), $$
we find that
$$ \EE\left[\left(U_n - \widehat{U}_n\right)^2\right] = O\left[\left(\frac{s(n)}{n}\right)^2\right], $$
where $\widehat{U}_n$ is defined as in \eqref{eq:hajek}.
\end{lemm}

We are now ready to prove our main result. As is common in extreme value theory, the result relies on a second-order convergence criterion. The full statement of Theorem \ref{theo:rbm} requires a non-degenerate second order condition ($\rho < 0$ as defined in \eqref{eq:second}), and does not necessarily hold in the limiting case with logarithmic second-order convergence ($\rho = 0$). As noted at the end of this section, however, we do not need this second-order condition to show that $\gh_{RBM}$ is consistent. For an overview of the second-order condition in extreme value theory, see e.g., \citet{dehaan2006extreme}.

\begin{theo}
\label{theo:rbm}
Let $X_1, ..., X_n$ be drawn \emph{iid} from a distribution $F$ satisfying the second-order condition
\begin{equation}
\label{eq:second}
\limt \frac{\frac{U(tx)}{U(t)} - x^\gamma}{A(t)} = x^\gamma \frac{x^\rho - 1}{\rho}, \ \forall x > 0
\end{equation}
for some $\gamma > 0$, $\rho < 0$, and a function $A(t) \rightarrow 0$ with constant sign. Here, $U(t)$ is the inverse quantile function
\begin{equation}
\label{eq:ut}
U(t) = \inf\left\{x : \frac{1}{1 - F(x)} \geq t\right\}.
\end{equation}
Moreover, suppose that $F$ satisfies the technical condition
\begin{equation}
\label{eq:technical}
\lim_{x \rightarrow 0} F(x) \cdot x^{-\frac{1}{\beta}} = 0 \ \text{for some} \ \beta > 0,
\end{equation}
and let $RBM_{k, n}$ be the RBM estimator as described in Lemma \ref{lemm:notation}.

If $k(n)$ is an intermediate sequence with $k(n) \rightarrow \infty$ and $k(n)/n \rightarrow 0$ such that
\begin{equation}
\label{eq:lambda}
\limn \sqrt{k(n)} A \left ( \frac{n}{k(n)} \right ) = \lambda \ \text{for some} \ \lambda \in \RR,
\end{equation}
then, for any $a > 0$, $RBM_{ak(n), n}$ is asymptotically normal with
\begin{equation}
\label{eq:RBM_k}
\sqrt{k(n)} (RBM_{ak(n), n} - \gamma) \Rightarrow \nn \left ( \lambda \Gamma(1 - \rho) \left(\frac{a}{2}\right)^{-\rho}, \frac{\gamma^2}{a} \right),
\end{equation}
where $\Gamma$ is the gamma function. Moreover, for any $a_1, ..., a_m > 0$, the estimators $RBM_{a_i k(n), n}$ are aymptotically jointly normal with covariance
$$ \limn k(n) \Cov[RBM_{a_i k(n), n}, RBM_{a_j k(n), n}] = \frac{2\gamma^2}{a_i + a_j}. $$
\end{theo}

The technical condition \eqref{eq:technical} is very weak, and can in practice be ignored. In an extreme value theoretic setup we usually care about very large values, whereas this condition only specifies the behavior of very small values. This condition trivially holds if $F$ is supported on $[\varepsilon, \infty)$ for some $\varepsilon > 0$.

We end this section by noting that, by Lemma \ref{lemm:hill}, even when $F$ does not satisfy the second-order condition for some $\rho < 0$, or when the sequence $k(n)$ does not satisfy \eqref{eq:lambda}, $H_1^{(2n/k(n))}$ still converges to $\gamma$ in expectation. Thus, by a slight modification of the proof of Theorem \ref{theo:rbm}, we find that, given any distribution $F$ with tail index $\gamma > 0$, $RBM_{k(n), n}$ is consistent for $\gamma$ along any intermediate sequence $k(n)$ provided $F$ satisfies the technical condition \eqref{eq:technical}.

\subsection{A Comparison with the Hill estimator}

\begin{figure}[t]
\centering
\includegraphics[width = 0.45\textwidth]{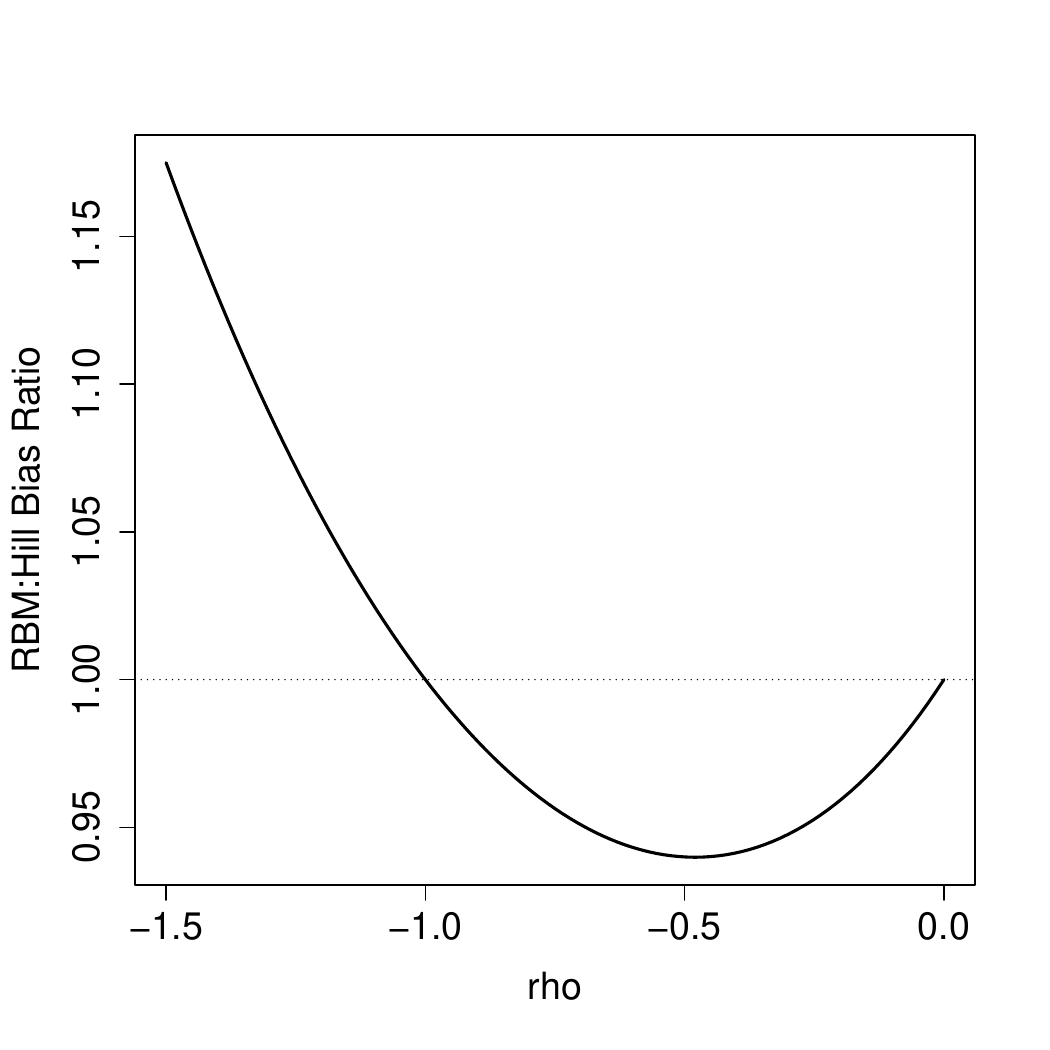}
\caption{Relative bias of the Hill and RBM estimators at equal variance, as a function of the second-order parameter $\rho$.}
\label{fig:bias}
\end{figure}

It is well known \citep[e.g.,][Theorem 3.2.5]{dehaan2006extreme} that under the conditions of Theorem \ref{theo:rbm}, the Hill estimator has an asymptotic distribution
\begin{equation}
\label{eq:H_k}
\sqrt{k(n)} (H_{ak(n), n} - \gamma) \Rightarrow \nn \left ( \frac{\lambda}{1 - \rho} a^{-\rho}, \frac{\gamma^2}{a} \right).
\end{equation}
At equal variance, the relative bias of the Hill and RBM estimators only depends on the second order parameter $\rho$:
$$ \limn \frac{\EE\left[RBM_{ak(n), n} - \gamma\right]}{\EE\left[H_{ak(n), n} - \gamma\right]} = 2^\rho \, \Gamma(1 - \rho) \, (1 - \rho). $$
We plot this function in Figure \ref{fig:bias}. For distributions with slow second-order convergence ($-1 < \rho < 0$) the RBM estimator is somewhat more efficient than the Hill estimator, but as we approach an idealized setting with very small $\rho$, the Hill estimator is less biased. In the common case of $\rho = -1$ (satisfied, e.g., by the Fr\'echet distribution), the Hill and RBM estimators have the same asymptotic bias.

Comparing \eqref{eq:RBM_k} with \eqref{eq:H_k}, we see that the RBM and Hill estimators have very similar behaviors when we only look at one value of $k$ at a time. The crucial difference between the two methods is that the RBM process $\gh(k)$ is a smooth function of $k$, while the Hill process is not.

\section{The RBM Process}
\label{sec:process}

Our result from the previous section leads naturally to the definition of an RBM process. Under the conditions of Theorem \ref{theo:rbm} with some $\gamma > 0$ and $\rho < 0$, let $k(n)$ be an intermediate sequence such that, for some finite $\lambda$,
$$ \limn \sqrt{k(n)}A\left(\frac{n}{k(n)}\right) = \lambda. $$
Then, writing
\begin{equation}
\label{eq:err}
X_n(t) = \sqrt{k(n)}\left(RBM_{tk(n), n} - \gamma\right),
\end{equation}
our result in Theorem \ref{theo:rbm} implies that, for all $t_1, ..., t_m > 0$, the $X_n(t_i)$ are asymptotically jointly normal with
\begin{align}
\label{eq:process}
&\limn \EE\left[X_n(t_i)\right] = \lambda \Gamma(1 - \rho) \left(\frac{t_i}{2}\right)^{-\rho} \\ \notag
&\limn \Cov\left[X_n(t_i), X_n(t_j)\right] = \frac{2\gamma^2}{t_i + t_j}.
\end{align}
These mean and covariance equations can be used to define a Gaussian process, which we call the RBM process.

\begin{defi}
\label{defi:process}
Given values $\gamma > 0$ and $\rho < 0$, let $X(t)$ be the Gaussian process on $(0, \, \infty)$ satisfying the mean and covariance relations \eqref{eq:process}. The RBM process $R(\tau)$ is then defined by $R(\tau) = X\left(e^\tau\right)$ for $\tau \in \RR$. 
\end{defi}

We define the RBM process on a log scale since this allows us to write down its properties more cleanly. This should not be too surprising, since the estimator as written in \eqref{eq:gh} is essentially a derivative $\frac{\partial M(s)}{\partial \log s}$ with $s$ on a log scale. In a similar vein, \citet{drees2000make} show that the Hill process is most naturally plotted with $k$ on a log scale. The following lemma shows that the $X_n(t)$ do in fact converge in law to the process $X(t)$.

\begin{lemm}
\label{lemm:convergence}

Let $X_n(t)$ be defined as in \eqref{eq:err} under the conditions of Theorem \ref{theo:rbm}, and let $X(t)$ be the auxiliary process from Definition \ref{defi:process} with the appropriate $\gamma > 0$ and $\rho < 0$. Then, the $X_n(t)$ converge weakly to $X(t)$ on compact intervals of $(0, \, \infty)$ under the Skorokhod topology on the space of cadlag functions $\mathcal{D}$.

\end{lemm}

Our RBM process is analogous to the Hill process as discussed in \citet{resnick1997smoothing}. These two processes, however, behave very differently. While the Hill process is equivalent to a modified Wiener process and so has continuous but non-differentiable sample paths, the RBM process has smooth sample paths.

\begin{theo}
\label{theo:process}
There exists a modification of the RBM process defined in \ref{defi:process} that has $\cc^\infty$ sample paths on $\RR$. Moreover, for any $\tau \in \RR$, $R(\tau)$ and its derivative $R'(\tau)$ have joint distribution
$$ \begin{pmatrix} R(\tau) \\ R'(\tau) \end{pmatrix}
\eqd \nn\left(\frac{\lambda \Gamma(1 - \rho) e^{-\rho \tau}}{2^{-\rho}} \begin{pmatrix} 1 \\ -\rho \end{pmatrix}, 
\frac{\gamma^2}{e^\tau} \begin{pmatrix} 1 & -1/2 \\ -1/2 & 1/2 \end{pmatrix}\right). $$
\end{theo}

In light of these results, we should expect the RBM estimator to have fairly smooth sample paths even for finite $n$. This is consistent with our observation in section \ref{sec:case} that the RBM estimator oscillates much less than either the Hill or the smooHill estimators.

\section{Threshold Selection}
\label{sec:opt}

Selecting a good tuning parameter $k$ for the Hill estimator is a classic problem in extreme value theory. Both the Hill and the RBM estimators have high variance at small $k$, and may be quite biased at high $k$. A successful choice of $k$ hinges on adequately balancing the bias and variance terms. Although the tuning parameter $k$ is integrated fairly differently in the Hill and RBM estimators, a given choice of $k$ has very similar effects on both estimators, and so our threshold selection heuristic should be read in light of the literature on threshold selection for the Hill estimator.

Most approaches to selecting $k$ require implicitly or explicitly estimating the second-order parameter $\rho$. \citet{danielsson2001using} and \citet{hall1990using} suggest using various sub-sample bootstraps to estimate the MSE-minimizing threshold in smaller samples. Transforming this small sample threshold into a full sample threshold, however, requires knowledge of $\rho$. \citet{hall1990using} recommends just using $\rho = -1$, while \citet{danielsson2001using} use auxiliary bootstraps to estimate the correct transformation coefficient.

\citet{drees1998selecting} suggest a procedure based on a law of the iterated logarithm, which also requires fitting $\rho$. Finally, \citet{beirlant2002exponential} advocate plugging a consistent estimator for $\rho$ into a formula for the optimal value of $k$ given by \citet{hall1985adaptive}.

An alternative approach to threshold selection aims to stop just before the smallest value of $k$ at which bias can be detected. \citet{hill1975simple} originally suggested picking $k$ just before the log spacings between consecutive order statistics fail a test for exponentiality. This test, however, was shown by \citet{hall1985adaptive} to be too lenient, and to produce estimates $\gh$ that were excessively biased. \citet{guillou2001diagnostic} remedy this problem by developing a way to jointly test for bias among high-order log spacings. The approach advocated by \citet{guillou2001diagnostic} does not require fitting $\rho$. This is a considerable benefit, since getting accurate estimates for $\rho$ is not practical in many applications.

We suggest a threshold selection rule for the RBM estimator that is similar in spirit to this second class of alternatives, in that it aims to select a threshold just before significant bias starts to appear. However, instead of stopping just before bias can be detected at a given significance level, we aim to minimize possible bias in a Bayesian sense.

As motivation for the proposed procedure, consider the RBM process $R(\tau)$ discussed in section \ref{sec:process}. From Theorem \ref{theo:process} we know that, if $\EE[R(\tau)] = b(\tau)$ is the bias at $\tau$, then
\begin{equation}
\label{eq:rp}
R'(\tau) \eqd \nn\left(-\rho b(\tau), \frac{\gamma^2}{2e^\tau}\right).
\end{equation}
This suggests using $R'(\tau)$ as a proxy for estimating bias. For heurstic motivation, suppose that for a fixed $\tau$, $b(\tau)$ is considered random with a uniform (improper) prior on $\RR$. Then, using \eqref{eq:rp}, we find that $b(\tau)$ has a posterior distribution
$$ \law\left[b(\tau) | R'(\tau)\right] \eqd \nn\left(\frac{R'(\tau)}{-\rho}, \frac{\gamma^2}{2\rho^2 e^\tau}\right), $$
and so
$$ \EE[b^2(\tau)|R'(\tau)] = \frac{2R'(\tau)^2 + \gamma^2e^{-\tau}}{2\rho^2}. $$
We then select
\begin{align}
\label{eq:opt}
\hat{\tau} &= \argmin_\tau \EE[b^2(\tau)|R'(\tau)] \\ \notag
&= \argmin_\tau R'(\tau)^2 + \frac{\gamma^2}{2e^\tau}.
\end{align}
We thus aim to select the value of $\tau$ that gives us least cause to suspect bias, rather than the first $\tau$ at which we must suspect bias. 

Although the threshold rule $\hat{\tau}$ was motivated fairly heuristically, it works well in our experiments. We begin our analysis of the threshold selection rule with a weak but important result.

\begin{theo}
\label{theo:thresh}
Let $R(\tau)$ be a RBM process satisfying \eqref{defi:process}, and let $\hat{\tau}$ be the threshold selected according to the rule described in \eqref{eq:opt}. Then
$\hat{\tau}$ is almost surely finite,
implying that the relative regret
\begin{equation}
\label{eq:thresh_rat}
\frac{R^2(\hat{\tau}) }{ \EE[R^2(\tau^*)]}
\end{equation}
of our adaptive procedure has a non-degenerate distribution, where $\tau^*$ is the optimal threshold, i.e., $\tau^* = \argmin_\tau \EE[R^2(\tau)]$.
\end{theo}

\begin{figure}[t]
\centering
\subfigure[Distribution of $\hat{\tau} - \tau^*$.]{
\includegraphics[width=0.45\textwidth]{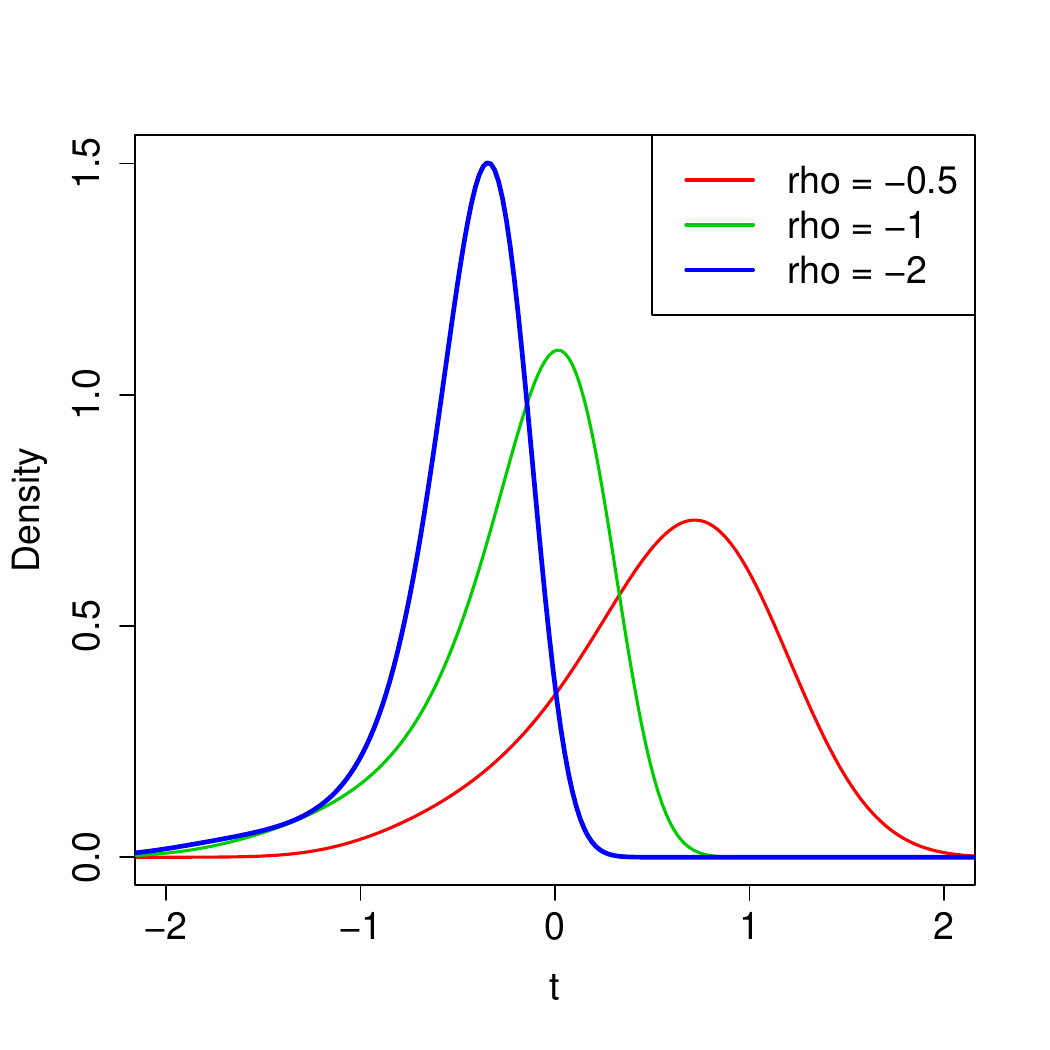}
\label{fig:tdist}
}
\subfigure[Relative regret of $\gh(\hat{\tau})$.]{
\includegraphics[width=0.45\textwidth]{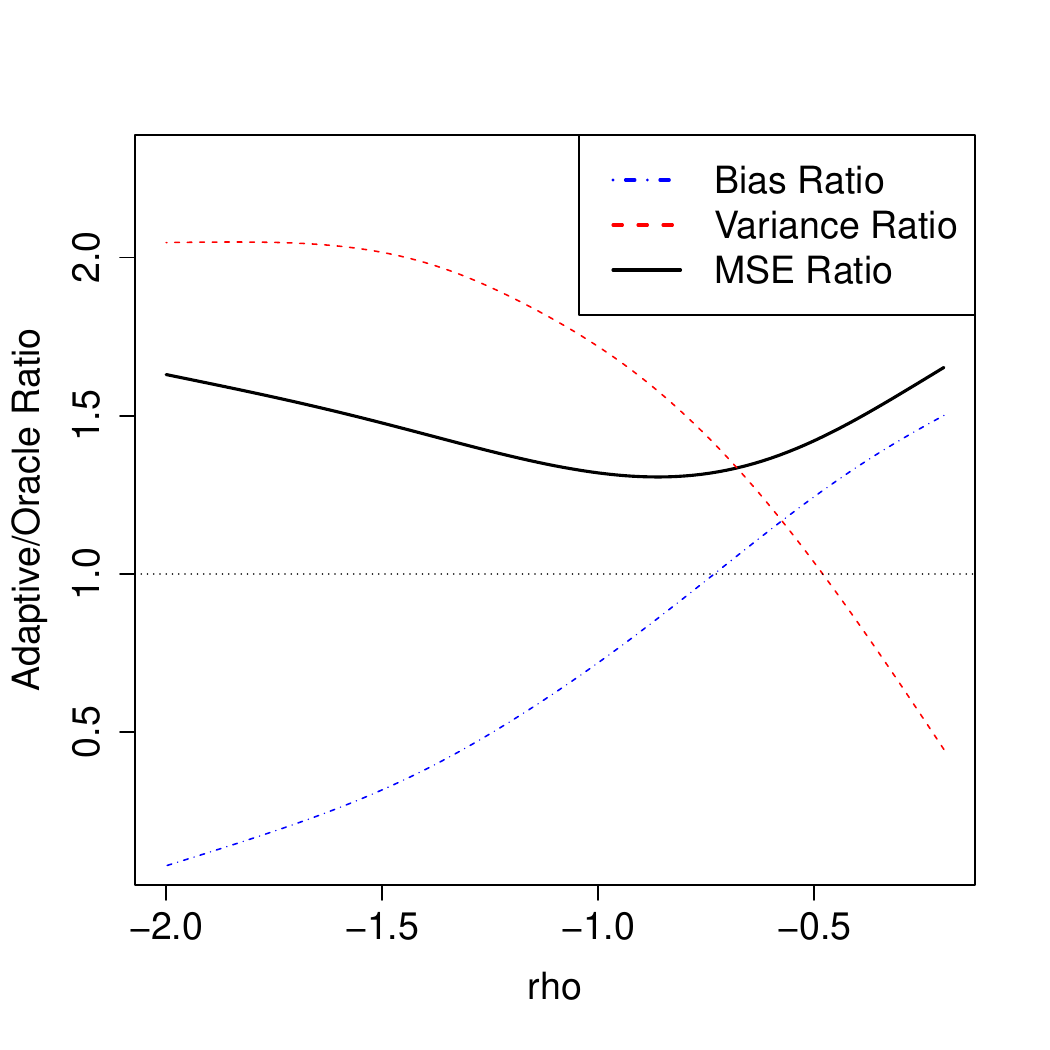}
\label{fig:MSE_rat}
}
\caption{Evaluating the threshold selection rule. The left panel shows the distribution of the adaptive threshold $\hat{\tau}$ relative to the optimal threshold $\tau^*$ as we vary the second-order parameter $\rho$. The right panel compares the loss from using the adaptive estimate $\gh(\hat{\tau})$  instead of the oracle $\gh(\tau^*)$. All simulations were run along the surface $\rho = -2\gamma$, and were scaled to $\lambda = 1$.}
\label{fig:thresh}
\end{figure}

We can get a better handle on the ratio in \eqref{eq:thresh_rat} by simulation. The behavior of $\hat{\tau}$ appears to depend largely on $\rho$. As shown in Figure \ref{fig:tdist}, $\hat{\tau}$ tends to be somewhat larger than the optimal value $\tau^*$ when $\rho$ is near 0, leading to lower variance but a slight excess in bias. The opposite is true when $\rho$ is very small. In Figure \ref{fig:MSE_rat}, we display the ratio from \eqref{eq:thresh_rat} along the surface $\rho = -2\gamma$.

Our optimization rule in \eqref{eq:opt} takes the form of an intuitive penalized optimization problem. Broadly speaking, the procedure tries to select a point $\tau$ such that $R_\tau'$ is small, since low $R_\tau'$ equates to low bias. However, low values of $\tau$ are plagued by high variability, and so we penalize small values of $\tau$. This procedure seems to mimic the strategy a practitioner might use in selecting $k$ from a Hill plot, and so we may hope that, even when the second-order condition does not hold or large third order effects are present, this rule will still give reasonable recommended thresholds.

We end this section on a note of caution: the relation \eqref{eq:rp} only holds in the tail region of the distribution. Thus, if we let $k$ grow large enough that the RBM estimator starts to use substantial amounts of non-tail data, our heuristic can fail badly.\footnote{ To witness such a failure, one can try applying the RBM estimator on 10'000 datapoints drawn from a Student-$t$ distribution with 2 degrees of freedom and a mean offset of +3. The rule from \eqref{eq:opt} will systematically pick a value of $k$ that is much too large.}
One way to avoid such a problem is, as discussed in section \ref{sec:case}, to pre-filter our data and to only give the RBM estimator datapoints that are in the tail area of the distribution. Theoretically, we should only give the RBM estimator the $n_0$ largest data points from a sample of size $n$, such that $n_0/n \rightarrow 0$ (e.g., take $n_0 = n^{0.9}$). In practice, we could decide to only use points that are above the mode of a coarse histogram of the data. Thankfully, such filtering should not cost us much, as the right-hand tail is the only part of the distribution that contains information that is relevant for estimating $\gamma$.

\section{Simulation Study}
\label{sec:simulation}

In this section, we run simulations to test our RBM estimator against three other estimators for $\gamma$. The benchmark estimators are all threshold selection rules for the Hill estimator, and are described in detail in \citet{beirlant2004statistics}. We compare

\begin{itemize}
\item $\gh_{RBM}$: Our RBM estimator, with threshold selection implemented as in \eqref{eq:rbmopt},
\item $\gh_{BDGS}$: The plugin method from \citet{beirlant2002exponential},
\item $\gh_{DK}$: The procedure based on a law of the iterated logarithm from \citet{drees1998selecting}, and
\item $\gh_{GH}$: The diagnostic for bias from \citet{guillou2001diagnostic}.
\end{itemize}

The distributions used for testing are given below. These distributions form a superset of the distributions used for a simulation study in \citet{beirlant2002exponential}.

\begin{itemize}
\item Fr\'echet(2) with distribution $F(x) = e^{-x^{-2}}$, $\gamma = 1/2$, $\rho = -1$. We drew $N = 200$ datapoints from this distribution.
\item Burr(1, 0.5, 2) with distribution $1 - (1 + \sqrt{x})^{-2}$, $\gamma = 1$, $\rho = -1/2$. We drew $N = 500$ datapoints from this distribution.
\item Student-$t(df)$ with $df = 3$ and 6 degrees of freedom, $\gamma = 1/df$, $\rho = -2/df$. We drew $N = 500$ datapoints from these distributions.
\item Log-Gamma(2, 1) with density $f(x) = x^{-2}\log(x)$, $\gamma = 1$, $\rho = 0$. We drew $N = 500$ datapoints from this distribution.
\item A distribution defined by $X = U^{-2} ( 1 - \log U)$ where $U$ is uniform on $[0, \, 1]$, with $\gamma = 2$ and $\rho = 0$. We drew $N = 500$ datapoints from this distribution.
\end{itemize}

Simulation results are given in Table \ref{tab:results}. All numbers were estimated using 4000 replications. Non-positive datapoints arising with the Student-$t$ distribution were discarded, as discussed in section \ref{sec:case}.

\begin{table}[t]
\caption{\label{tab:results}Comparison of root mean squared error (RMSE) and bias for four estimators. Standard sampling errors ($\times 10^{-3}$) are indicated in parentheses.}
\centering
\begin{tabular}{*{2}{r}|*{4}{l}|}
\hline
Distribution & & $\gh_{RBM}$ & $\gh_{BDGS}$ & $\gh_{DK}$ & $\gh_{GH}$ \\ \hline
Fr\'echet & RMSE & 0.116 (2) & 0.142 (4) & {\bf 0.087} (1) & 0.102 (1) \\
 & Bias & 0.011 (2) & {\bf -0.004} (2) & 0.035 (1) & 0.044 (1) \\ \hline
Burr & RMSE & {\bf 0.334} (3) & 0.442 (3) & 0.344 (3) & 0.382 (3) \\
 & Bias & {\bf 0.129} (5) & 0.410 (3) & 0.261 (4) & 0.333 (3) \\ \hline
Student-$t$(3) & RMSE & 0.119 (2) & {\bf 0.113} (2) & 0.12 (2) & 0.145 (2) \\ 
   & Bias & {\bf 0.034} (4) & 0.08 (3) & 0.089 (3) & 0.123 (2) \\  \hline
Student-$t$(6) & RMSE & {\bf 0.112} (1) & 0.145 (1) & 0.149 (1) & 0.178 (1) \\
 & Bias & {\bf 0.074} (1) & 0.130 (1) & 0.134 (1) & 0.168 (1) \\ \hline
Log-Gamma & RMSE & 0.293 (2) & {\bf 0.258} (3) & 0.327 (2) & 0.287 (2) \\
 & Bias & 0.215 (3) & {\bf 0.182} (3) & 0.301 (2) & 0.238 (3) \\ \hline
$U^{-2} ( 1 - \log U)$ & RMSE & 0.434 (4) & {\bf 0.399} (8) & 0.452 (4) & 0.404 (6) \\ 
   & Bias & 0.363 (5) & {\bf 0.218} (7) & 0.416 (4) & 0.284 (7) \\ \hline
\end{tabular}
\end{table}

We see that the RBM estimator is overall competitive with the other tested estimators in terms of MSE: RBM performs particularly well for both the Burr and the Student-$t$, and behaves reasonably for the rest. The main exception to this pattern are the last two distributions with $\rho = 0$, for which $\gh_{BDGS}$ performs well. 

We note in particular that $\gh_{RBM}$ is less biased than either $\gh_{GH}$ or $\gh_{DK}$ for the surveyed distributions. At equal MSE, having low bias may be advantageous since variance terms are often easier to estimate than bias terms which depend on second-order parameters, and since systematic bias across multiple experiments may be hard to detect.

\section{Conclusion}
\label{sec:conclusions}

In this paper, we presented a new estimator for the tail index of a distribution in the Fr\'echet domain of attraction. The estimator arose from studying the maxima of randomly generated subsamples, but can also be described as an infinite order $U$-statistic taken over the Hill estimator. The main advantage of our RBM estimator in comparison with existing methods lies in its stability and ease of use. While most commonly used estimators are extremely sensitive to small changes in the tuning parameter $k$, the RBM estimator is stable with respect to $k$. And, while most other estimators require either manually choosing the threshold or fitting a complicated auxiliary model for $k$, the RBM framework admits a simple, intuitive, and largely automatic heuristic for threshold selection. Although the results proved in this paper are asymptotic, we saw in section \ref{sec:case} that the advantages of the RBM estimator are apparent in finite samples.

More generally, this paper presents a new approach to constructing and finding the limiting distribution of tail index estimators. The asymptotic behavior of many classical estimators can be established using results from, e.g., \citet{drees1998smooth} on the convergence of tail empirical processes. In the present work, however, we took a different approach and studied convergence using H\'ajek projections and infinite order $U$-statistics. There are multiple opportunities to tackle further problems in extreme value theory using similar methods. In particular, it should be possible to construct a bias-corrected version of the RBM estimator by mirroring ideas from \citet{gomes2008tail}, to establish an RBM-type process which would permit estimation of a general tail index $\gamma \in \RR$, and to use similar subsampling ideas in the study of multivariate extremes.

\section*{Acknowledgment}

I am grateful to Val\'erie Chavez-Demoulin, Bradley Efron, Jonathan Taylor, Julie Tibshirani, Suzanne de Treville, Guenther Walther and two anonymous referees for comments and helpful suggestions. This work was supported by a B. C. and E. J. Eaves Stanford Graduate Fellowship.

\bibliographystyle{plainnat}
\bibliography{references}

\newpage
\section{Appendix: Proofs}
\label{sec:appendix}

In the following results, we use the notation $U(t)$ for the inverse quantile function as defined in \eqref{eq:ut}.
It can be shown \citep[e.g.,][section 1.2]{dehaan2006extreme} that the distribution $F$ has extreme value index $\gamma > 0$ if and only if $U$ is a regularly varying function of index $\gamma$, i.e. $\lim_{t \rightarrow \infty} U(tx)/U(t) = x^\gamma$ for all $x > 0$. We also use $[n]$ for the set $\{1, ..., n\}$, and
$ \EE[X; A] = \int X 1_A \, d\PP $
where $A$ is a set and $1_A$ its indicator function.

\subsection{Preparatory Lemmas}

\begin{lemm}
\label{lemm:hill}

Let $X_1, ..., X_s$ be drawn \emph{iid} from a distribution $F$ of strictly positive support with extreme value index $\gamma > 0$. Then the first Hill estimator $H_1^{(s)}$ from \eqref{eq:h1} converges in distribution to an exponential random variable with mean $\gamma$. Moreover, if there is a constant $\beta > 0$ such that \eqref{eq:technical} holds,
then all moments of $H_1^{(s)}$ converge to the corresponding moments of the limiting random variable. In particular,
$$ \lims \EE\left[H_1^{(s)}\right] = \gamma \text{ and } \lims \Var\left[H_1^{(s)}\right] = \gamma^2. $$

\proof

In terms of the inverse quantile function $U(t)$ from \eqref{eq:ut}, we can write $X_k \eqd U(Y_k)$, where the $Y_k$ are drawn independently from a distribution with cdf $F_Y(y) = \frac{y - 1}{y}$ for $y > 1$. We write $Y_{1, s} \leq ... \leq Y_{s, s}$ for the order statistics of the $Y_k$.

Since $U$ is a regularly varying function of index $\gamma$, Potter's inequality \citep{potter1942mean} implies that, for any $\varepsilon > 0$, there is a $t_0$ such that, for all $t, \ tx \geq t_0$,
\begin{equation}
\label{eq:potter}
(1 - \varepsilon) x^{\gamma - \sgn[\log x] \cdot \varepsilon} < \frac{U(tx)}{U(t)} < (1 +  \varepsilon) x^{\gamma + \sgn[\log x] \cdot \varepsilon},
\end{equation}
where $\sgn$ is the sign operator. Thus, since as in Lemma \ref{lemm:small} $\lims \PP[Y_{s - 1, s} < t_0] = 0$ and since, as shown below, the terms $Y_{s, \, s} / Y_{s-1, \, s}$ are uniformly bounded, we conclude that
$$\log \left[\frac{U(Y_{s,s})}{U(Y_{s-1, s})}\right] - \gamma \cdot \log\left[\frac{Y_{s,s}}{Y_{s-1, s}}\right] \rightarrow_p 0. $$
Now, we note that the $\log Y_k$ have standard exponential distribution $Exp(1)$. By R\'enyi representation \citep{renyi1953theory}, if $E_{1, s} \leq ... \leq E_{s, s}$ are order statistics of a standard exponential distribution, the $E_{k, s}$ are jointly distributed as
\begin{equation}
\label{eq:renyi}
E_{k, s} \eqd \sum_{l = 1}^k \frac{E_l^*}{s - l + 1}, \text{ with } E_1^*, \ ..., \ E_s^* \sim Exp(1).
\end{equation}
In particular, $E_{s, s} - E_{s - 1, s}$ is exponentially distributed, and is independent from $E_{s - 1, s}$. This implies our first claim:
$$ \gamma \cdot \log\left[\frac{Y_{s,s}}{Y_{s-1, s}}\right] \eqd Exp(\gamma), \text{ and so } H_1^{(s)} \eqd   \log \left[\frac{U(Y_{s,s})}{U(Y_{s-1, s})}\right] \Rightarrow Exp(\gamma). $$
To show convergence of the $\nu^{th}$ moment, we again use Potter's inequality, which implies that for any $\varepsilon > 0$ there is a $t_0$ such that
\begin{align*}
p_s(t_0) \cdot \EE\left[\left(\log[1 - \varepsilon] + (\gamma - \varepsilon) \cdot E\right)^{\nu}\right]
&< \EE\left[\log \left[\frac{U(Y_{s,s})}{U(Y_{s-1, s})}\right]^\nu; Y_{s-1, s} > t_0\right] \\
&<p_s(t_0) \cdot \EE\left[\left(\log[1 + \varepsilon] + (\gamma + \varepsilon) \cdot E\right)^{\nu}\right],
\end{align*}
where $E \sim Exp(1)$ and $p_s(t_0) = \PP[Y_{s-1, s} > t_0]$. We recall that $p_s(t_0) \rightarrow 1$, and so in order to obtain convergence of moments it suffices to show that 
$$ \lims \EE\left[\log \left[\frac{U(Y_{s,s})}{U(Y_{s-1, s})}\right]; Y_{s-1, s} < t_0\right] = 0; $$
this follows from the second part of Lemma \ref{lemm:small}, since the technical condition near 0 holds by hypothesis.
\endproof

\end{lemm}

\begin{lemm}
\label{lemm:bias}
Let $X_1, ..., X_s$ be drawn from a distribution $F$ satisfying the second-order condition \eqref{eq:second}
for all $x > 0$, with some $\gamma > 0$, $\rho < 0$ and a positive or negative function $A(t)$ with $\limt A(t) = 0$. Moreover, suppose there is a constant $\beta > 0$ such that \eqref{eq:technical} holds.
Then, writing $X_{1, s} \leq ... \leq X_{s,s}$ for the order statistics of $X$, we have, for any $\alpha > 0$, that
$$ \lims \frac{\EE[\log X_{s,s} - \log X_{s-1, s}] - \gamma}{A(\alpha s)} =  \frac{\Gamma(1 - \rho)}{\alpha^\rho}. $$
\proof
As in the proof of Lemma \ref{lemm:hill}, we write $X_k \eqd U(Y_k)$ where the $Y_k$ have cdf $F_Y(y) = \frac{y - 1}{y}$ for $y \geq 1$. Since $A(t) \rightarrow 0$, the stated second-order condition is equivalent to
$$\limt \frac{\log U(tx) - \log U(t) - \gamma \log(x)}{A(t)} = \frac{x^\rho - 1}{\rho}$$
for all $x > 0$. By \citet{drees1998smooth}, there exists a function $A_0(t) \sim A(t)$ (and so without loss of generality $A_0(t) = A(t)$) such that for any $\varepsilon > 0$, there is a $t_0$ such that, for all $t > t_0$ and $x \geq 1$,
\begin{equation}
\label{eq:drees}
\left|\frac{\log U(tx) - \log U(t) - \gamma \log(x)}{A(t)} - \frac{x^\rho - 1}{\rho}\right| < \varepsilon x^{\rho + \varepsilon}.
\end{equation}
For any $ r < 1$ we find by R\'enyi representation \eqref{eq:renyi} that
$$ \EE\left[\left(\frac{Y_{s,s}}{Y_{s-1,s}}\right)^{r} \big| Y_{s-1, s} \geq t_0 \right] = \int_0^\infty e^{(r - 1)x} \ dx < \infty, $$
and so, because $\rho < 0$, we find by plugging $t = Y_{s-1,s}$ and $tx = Y_{s,s}$ into \eqref{eq:drees} that for any $\delta > 0$ there is a $t_0$ such that
\begin{equation}
\label{eq:bigg}
 \lims  \EE\left[\left(\frac{\log \left[\frac{U(Y_{s,s})}{U(Y_{s-1,s})}\right] - \gamma \log\left(\frac{Y_{s,s}}{Y_{s-1,s}}\right)}{A(Y_{s-1,s})} - \frac{\left(\frac{Y_{s,s}}{Y_{s-1,s}}\right)^\rho - 1}{\rho}\right)^2; Y_{s-1, s} \geq t_0\right] < \delta.
 \end{equation}
We now move to the case $Y_{s-1, s} < t_0$. $A(t)$ must be regularly varying \citep[e.g.,][section 2.3]{dehaan2006extreme} with index $\rho$, and so by Karamata representation we can assume without loss of generality that $A(t)$ is continuous on $[0, \infty)$ and strictily positive or strictly negative; in particular, $A(t)$ is then bounded away from 0 for finite intervals. Thus, by Lemma \ref{lemm:small}, the expression on \eqref{eq:bigg} now integrated over the set $Y_{s-1, s} < t_0$ converges to 0. From this we conclude that
$$ \lims  \EE\left[\left(\frac{\log \left[\frac{U(Y_{s,s})}{U(Y_{s-1,s})}\right] - \gamma \log\left(\frac{Y_{s,s}}{Y_{s-1,s}}\right)}{A(Y_{s-1,s})} - \frac{\left(\frac{Y_{s,s}}{Y_{s-1,s}}\right)^\rho - 1}{\rho}\right)^2\right] = 0. $$
Moreover, assuming without loss of generality that appropriate regularity conditions for $A(t)$ hold near $t = 0$, we can show along the lines of Lemma \ref{lemm:hill} that
$$ \limsup_{s \rightarrow \infty} \EE\left[\left(\frac{A(Y_{s-1,s})}{A(s)}\right)^2\right] < \infty,
\ \text{and} \
\lims \EE\left[\left(\frac{A(Y_{s-1,s})}{A(s)} - \left(\frac{Y_{s-1, s}}{s}\right)^\rho \right)^2\right] = 0. $$
Thus, using Cauchy-Schwarz, we establish that
$$ \lims  \EE\left[\frac{\log \left[\frac{U(Y_{s,s})}{U(Y_{s-1,s})}\right] - \gamma \log\left(\frac{Y_{s,s}}{Y_{s-1,s}}\right)}{A(s)} - \frac{\left(\frac{Y_{s,s}}{s}\right)^\rho - \left(\frac{Y_{s-1,s}}{s}\right)^\rho}{\rho}\right] = 0. $$
Finally, by R\'enyi representation we can write
\begin{align*}
\left(\frac{Y_{s, s}}{s}, \frac{Y_{s-1,s}}{s}\right)
&\eqd \left(\frac{s^{-1}}{1 - \exp[-\tilde{E}_{1,s}]},\frac{s^{-1}}{1 - \exp[-\tilde{E}_{2,s}]}\right) \\
&\Rightarrow \left(\frac{1}{E_1}, \frac{1}{E_1 + E_2}\right),
\end{align*}
where $E_1$ and $E_2$ are independent standard exponential and the $\tilde{E}_{k, s}$ are exponential order statistics.  Uniform integrability holds, and so
$$\lims \frac{\EE\log \left[\frac{U(Y_{s,s})}{U(Y_{s-1,s})}\right] - \gamma \EE\log\left[\frac{E_1 + E_2}{E_1}\right]}{A(s)} = \EE\left[\frac{\left(\frac{1}{E_1}\right)^\rho - \left(\frac{1}{E_1 + E_2}\right)^\rho}{\rho} \right].$$
Writing $f_{\chi^2_p}$ for the density of the chi-squared distribution with $p$ degrees of freedom, the right-hand side expectation is
\begin{align*}
\int_0^\infty \frac{2f_{\chi^2_2}\left(2x\right)}{\rho \, x^\rho} \, dx - \int_0^\infty \frac{2f_{\chi^2_4}(2x)}{\rho \, x^\rho} \, dx 
&= \frac{1}{\rho}\Gamma(1 - \rho) - \frac{1}{\rho}\Gamma(2 - \rho) \\
&= \Gamma(1 - \rho).
\end{align*}
The desired conclusion follows by recalling that $A(t)$ is regularly varying of index $\rho$.
\endproof
\end{lemm}

\begin{lemm}
\label{lemm:variance}
Let $X_{1, s} \leq ... \leq X_{s, s}$ be independent order statistics drawn from a distribution $F$ with extreme value index $\gamma > 0$, satisfying \eqref{eq:technical} for some $\beta > 0$.
Then, writing
$$ \Psi_s(X) = \EE[\log X_{s, s} - \log X_{s-1, s} | X_1 = X], $$
we have:
\begin{align*}
&\lims s \Var[\Psi_s] = \frac{\gamma^2}{2}, \ \text{and, more generally,} \\
&\lims s \Cov[\Psi_s, \Psi_{\alpha s}] = \frac{\gamma^2}{1 + \alpha}
\end{align*}
for all $\alpha > 0$.
\proof
For convenience, write $W_i = \log X_i$. For $W_1, ..., W_s$, and $\tilde{W}_1, ..., \tilde{W}_{s+1}$ independent of each other,
\begin{align*}
\delta_s(\tilde w)
:&= \EE[\tilde{W}_{s+1, s+1} - \tilde{W}_{s, s+1}|\tilde{W}_1 = \tilde w] - \EE[{W}_{s,s} - {W}_{s-1, s}] \\
&= \EE[W_{s-1, s} - \tilde w; W_{s-1, s} < \tilde w < W_{s,s}] + \EE[\tilde w - 2W_{s,s} + W_{s-1,s}; W_{s,s} < \tilde w] \\
&= \EE[W_{s-1,s}; W_{s-1,s} < \tilde w] - 2\EE[W_{s,s}; W_{s,s} < \tilde w] \\
& \ \ \ + \tilde w \cdot \left(2\PP[W_{s,s} < \tilde w] - \PP[W_{s-1,s} < \tilde w]\right).
\end{align*}
Our goal is to study the distribution of $\delta_s(\log X)$ when $X \sim F$. We now proceed by evaluating each of these terms separately. As in the proof of Lemma \ref{lemm:bias},
\begin{align*}
&W_{s,s} - \log U(s) \Rightarrow - \gamma \log E_1 \ \text{and} \\
&W_{s-1,s} - \log U(s) \Rightarrow - \gamma \log (E_1 + E_2),
\end{align*}
where the $E_i$ are independent standard exponential random variables.

We can use the Potter bounds \eqref{eq:potter} and Lemma \ref{lemm:small} to show that the sequences $W_{s,s} - \log U(s)$ and $W_{s-1,s} - \log U(s)$ are uniformly integrable. This enables us to find the moments of interest from the limiting distributions. First, for all $w \in \RR$,
\begin{align*}
\lims &\EE[W_{s,s} - \log U(s); W_{s,s} - \log U(s) < w] \\
&= -\int_{e^\frac{-w}{\gamma}}^\infty \gamma \log (x) \cdot e^{-x} \ dx \\
&= w e^{-e^\frac{-w}{\gamma}} - \gamma \Gamma\left(0, e^\frac{-w}{\gamma}\right),
\end{align*}
where $\Gamma$ is the partial gamma function. Similarly,
\begin{align*}
\lims &\EE[W_{s-1,s} - \log U(s);W_{s-1,s} - \log U(s) < w] \\
&= - \int_{e^\frac{-w}{\gamma}}^\infty \gamma \log (x) \cdot x e^{-x} \ dx \\
&= w (1 + e^\frac{-w}{\gamma}) e^{-e^\frac{-w}{\gamma}} - \gamma \left[ e^{-e^\frac{-w}{\gamma}} + \Gamma\left(0, e^\frac{-w}{\gamma}\right)\right].
\end{align*}
Finally,
\begin{align*}
\lims &2\PP[W_{s,s} - \log U(s) < w] - \PP[W_{s-1,s} - \log U(s) < w] \\
&= 2\PP\left[E_1 > e^\frac{-w}{\gamma}\right] - \PP\left[E_1 + E_2 > e^\frac{-w}{\gamma}\right] \\
&= \left(1 - e^\frac{-w}{\gamma}\right)e^{-e^\frac{-w}{\gamma}}.
\end{align*}
Combining all our expressions, we find that
\begin{equation}
\label{eq:to_integrate}
\lims \delta_s(w + \log U(s)) = \gamma \cdot \left [\Gamma\left(0, e^\frac{-w}{\gamma}\right) - e^{-e^\frac{-w}{\gamma}} \right].
\end{equation}
It remains to find the distribution of
$$ z_s := \exp\left[ -\frac{\log X - \log U(s)}{\gamma}\right], $$
when $X$ is drawn from $F$. Now,
\begin{align*}
\lims s\PP[z_s < \chi]
&= \lims s \PP \left [X > U(s)\chi^{-\gamma} \right] \\
&= \lims s \PP \left [X > U \left (s\chi^{-1} \right) \right] \\
&= \chi,
\end{align*}
for any $\chi > 0$. Thus, if $\mu_{z_s}$ is the distribution of $z_s$, we find that $s \cdot \mu_{w_s}$ converges weakly to Lebesgue measure on compact intervals of $\RR_+$.

Now, by construction, we see that the functions $g_s(w) = \delta_s(w + \log U(s))$ must be Lipshitz continuous with constant 1 (since changing $\tilde{W}_1$ by $\Delta$ can change $\tilde{W}_{s + 1,s} - \tilde{W}_{s,s}$ by at most $\Delta$), and so the $g_s$ converge uniformly on compact intervals to $g$, where $g(w) = f\left(e^\frac{-w}{\gamma}\right)$ and
$$ f(z) = \gamma \cdot \left [\Gamma\left(0, z\right) - e^{-z} \right] $$
is the limiting function from \eqref{eq:to_integrate}. We can then argue by weak convergence of the $\mu_{z_s}$ to Lebesgue measure and by uniform convergence of $|g_s - g|$ to 0 that:
\begin{align*}
&\limc \lims s\EE\left[\delta_s(\log X); \left|\log\left[\frac{X}{U(s)}\right]\right| \leq c\right] = \limc \int_{e^{-c/\gamma}}^{e^{c/\gamma}} f(z) \ dz  = 0, \ \text{and}\\
&\limc \lims s\EE\left[\delta^2_s(\log X); \left| \log\left[\frac{X}{U(s)}\right] \right| \leq c\right] = \limc \int_{e^{-c/\gamma}}^{e^{c/\gamma}} f^2(z) \ dz = \frac{\gamma^2}{2}.
\end{align*}
It now remains to show uniform integrability of $s\delta_s^2$. Consider the residuals
$$ R_c = \lims s\EE\left[\delta^2_s(\log X); \left| \log\left[\frac{X}{U(s)}\right] \right| > c\right]. $$
By dominated convergence, if any one of the $R_c$ is finite, then $\limc R_c = 0$. Thus, the $R_c$ only have two possible limiting values: 0 or infinity. Now, by Hoeffding's inequality \citep{hoeffding1948class}, we know that
$$ s \Var[\Psi_s] \leq \Var[H_1^{(s)}] $$
for all $s$. Moreover, from Lemma \ref{lemm:hill}, we know that $\Var[H^{(s)}_1] \rightarrow \gamma^2$. Thus,
$$s \EE[\delta_s^2 (\log X)] = s \Var[\Psi_s] \leq \gamma^2$$
 for all $s$. This implies that the $R_c$ are also bounded by $\gamma^2$, and so must converge to zero; thus our stated result about variance holds.

More generally, for any $\alpha > 0$, we find that
\begin{align*}
\lims &s \Cov[\Psi_s, \Psi_{\alpha s}] \\
&= \gamma^2 \int_0^\infty \left (\Gamma(0, x) - e^{-x} \right) \cdot \left (\Gamma(0, \alpha x) - e^{-\alpha x} \right) \ dx \\
&= \gamma^2 \left[x\Gamma(0, x)\Gamma(0, \alpha x) - \frac{e^{-(\alpha + 1)x}}{\alpha + 1}\right]_{x = 0}^\infty \\
&= \frac{\gamma^2}{\alpha + 1}.
\end{align*}
\endproof
\end{lemm}

\subsection{Proof of Main Results}

\begin{proof}[Proof of Lemma \ref{lemm:notation}]
Using our notation from \eqref{eq:max},
\begin{align*}
RBM_{k, n}
= &\ s \cdot \left(M(s) - M(s-1)\right) \\
= &\ s \cdot \binom{n}{s}^{-1} \sum_{\{A \subseteq [n]: |A| = s\}} \log\left(\max_{a \in A}\{X_a\}\right) \\
  &\ \ \ - s \cdot \frac{n - s + 1}{s} \cdot \binom{n}{s}^{-1} \sum_{\{B \subseteq [n]: |B| = s - 1\}} \log\left(\max_{b \in B}\{X_b\}\right) \\
= &\ \binom{n}{s}^{-1} \sum_{\{A \subseteq [n]: |A| = s\}} \Bigg[ s \cdot \log\left(\max_{a \in A}\{X_a\}\right) \\
  &\ \ \ - \sum_{\{B \subset A: |B| = s-1\}} \log\left(\max_{b \in B}\{X_b\}\right) \Bigg] \\
= &\ \binom{n}{s}^{-1} \sum_{\{A \subseteq [n]: |A| = s\}} H_1^{(s)}(X_A),
\end{align*}
To obtain the second-to-last line, we used the fact that each set $B$ of size $s-1$ is a subset of $n-s+1$ distinct sets of size $s$.
\end{proof}

\begin{proof}[Proof of Lemma \ref{lemm:anova}]
Without loss of generality, we can assume that the $g^{(n)}$ all have zero mean. By the Efron-Stein ANOVA decomposition \citep{efron1981jackknife}, for each $g^{(n)}$, there exist $j$-parameter symmetric functions $G_j^{(n)}$ with $j = 1, ..., s(n)$ such that
$$ g^{(n)}(X_1, ..., X_{s(n)}) = \sum_{j = 1}^{s(n)} \sum_{\{I_j \in [s(n)]:|I_j| = j\}} G_j^{(n)}\left(X_{I_j}\right), $$
and the $G_j^{(n)}\left(X_{I_j}\right)$ are all mean-zero and uncorrelated. Using this result, we can write our $U$-statistic as
$$ U_n = \binom{n}{s(n)}^{-1} \sum_{j = 1}^{s(n)} \binom{n - j }{ s(n) - j} \sum_{\{I_j \in [n]: |I_j| = j\}} G_j^{(n)}\left(X_{I_j}\right). $$
Moreover, under this notation,
$$G_1^{(n)}(X_1) = \EE\left[g^{(n)}(X_1, X_2, ..., X_{s(n)})|X_1\right],$$
and
\begin{equation}
\label{eq:rbm_to_hajek}
\widehat{U}_n = \binom{n}{s(n)}^{-1} \binom{n - 1 }{ s(n) - 1} \sum_{i = 1}^n G_1^{(n)}\left(X_i\right).
\end{equation}
Thus since the $G_j$ are uncorrelated and the $X_i$ are \emph{iid},
\begin{align*}
\EE\left[\left(U_n - \widehat{U}_n\right)^2\right]
&= \binom{n}{s(n)}^{-2} \sum_{j = 2}^{s(n)} \binom{n - j }{ s(n) - j}^2 \binom{n }{ j} \Var\left[G_j^{(n)}\right] \\
&\leq \frac{s(n)(s(n) - 1)}{n(n-1)} \sum_{j = 2}^{s(n)} \binom{s }{ j} \Var\left[G_j^{(n)}\right] \\
&\leq \frac{s(n)(s(n) - 1)}{n(n-1)} \Var\left[g^{(n)}\right],
\end{align*}
which implies the stated result, since $\Var\left[g^{(n)}\right] \leq C$ by hypothesis.
\end{proof}

\begin{proof}[Proof of Theorem \ref{theo:rbm}]
Let $s(n) =  2n/k(n) $ be the subsample block size. By Lemma \ref{lemm:hill},
$$ \lim_{s \rightarrow \infty} \Var\left[H_1^{(s)}\right] = \gamma^2. $$
Thus, by Lemma \ref{lemm:anova}, $RBM_{k(n), n}$ converges in mean square to its H{\'a}jek projection $\widehat{RBM}_{k(n), n}$, and
$$ \limn k(n) \cdot \EE\left[\left(\widehat{RBM}_{k(n), n} - RBM_{k(n), n}\right)^2\right] = 0, $$
because $k(n) \cdot \left(\frac{s(n)}{n}\right)^2 \sim 4/k(n)$ converges to zero. Moreover, as in \eqref{eq:rbm_to_hajek}, for any $a > 0$ we can write this projection as
\begin{equation}
\label{eq:rbm_proj}
\widehat{RBM}_{ak(n), n} = \frac{s(n)}{an} \sum_{i = 1}^n \EE\left[H_1^{(s(n)/a)}|X_i\right] - \left(\frac{s(n)}{a} - 1\right)\EE\left[H_1^{(s(n)/a)}\right].
\end{equation}
From Lemmas \ref{lemm:bias} and \ref{lemm:variance}, we get that for any $a, b > 0$,
\begin{align*}
& \limn \frac{\EE[H_1^{(s(n)/a)}] - \gamma}{A(s(n)/2)} = \left(\frac{a}{2}\right)^{-\rho}\Gamma(1 - \rho), \ \text{and} \\
& \limn s(n) \Cov\left[\EE\left(H_1^{(s(n)/a)}|X_1\right), \EE\left(H_1^{(s(n)/b)}|X_1\right)\right] = \frac{\gamma^2}{a^{-1} + b^{-1}},
\end{align*}
the second of which implies, together with \eqref{eq:rbm_proj}, that
\begin{equation*}
\label{eq:hajek_cov}
\limn k(n) \Cov\left[\widehat{RBM}_{ak(n), n}, \widehat{RBM}_{bk(n), n}\right] = \frac{2\gamma^2}{a + b}.
\end{equation*}
With these expressions in hand, we can conclude using the central limit theorem for triangular arrays and Slutsky's lemma that $RBM_{k(n), n}$ has the stated asymptotic distribution.
\end{proof}

\begin{proof}[Proof of Lemma \ref{lemm:convergence}.]
We already know from Theorem \ref{theo:rbm} that the finite dimensional distributions of the $X_n(t)$ converge in law to those of $X(t)$. Thus, to show that $X_n(t) \Rightarrow X(t)$ in $\mathcal{D}_{[a,b]}$ for some $0 < a < b$, it suffices by, e.g., Theorem 15.6 of \citet{billingsley1968convergence} to show that there is a constant $C$ such that, given any $\varepsilon > 0$, there exists a constant $N_\varepsilon$ such that for all $t_1, t_2 \in [a,b]$ with $|t_1 - t_2| < \varepsilon$ and for all $n \geq N_\varepsilon$,
\begin{equation}
\label{eq:moment_bound}
\EE\left[\left(X_n(t_2) - X_n(t_1)\right)^2\right] \leq C \varepsilon^2.
\end{equation}
To show such a bound, it is useful to decompose our expression:
\begin{align*}
\EE\left[\left(X_n(t_2) - X_n(t_1)\right)^2\right] \leq
\ &\EE\left[\left(X_n(t_2) - \widehat{X}_n(t_2)\right)^2\right] \\
&+\Var\left[\widehat{X}_n(t_2) - \widehat{X}_n(t_1)\right] \\
&+\EE\left[X_n(t_2) - X_n(t_1)\right]^2 \\
&+\EE\left[\left(\widehat{X}_n(t_1) - X_n(t_1)\right)^2\right], \\
\end{align*}
where $\widehat{X}$ is a H{\'a}jek projection of $X$ as defined in \eqref{eq:hajek}. It now remains to bound the terms individually.

By Lemma \ref{lemm:anova}, the first and the last summands decay uniformly as $O(1/k(n))$ on $[a, b]$, and so become eventually negligible for any $\varepsilon > 0$. Meanwhile, as in \eqref{eq:rbm_proj}, we can write the variance term (i.e. the second summand) as
$$ \frac{4n}{t_1 t_2 k(n)}
\Var\left[\EE\left(H_1^{\left(\frac{2n}{k(n)t_2}\right)} - H_1^{\left(\frac{2n}{k(n)t_1}\right)}|X_1\right)\right], $$
which by Lemma \ref{lemm:variance} converges to
$$ \frac{\gamma^2(t_2 - t_1)^2}{t_1t_2(t_1 + t_2)} \leq \frac{\gamma^2}{2a^3} (t_2 - t_1)^2 $$
on $[a, b]$; the result can be extended to show that the convergence is uniform over the interval. Finally, Lemma \ref{lemm:bias} reduces the problem of showing that $\EE[X_n(t)]$ satisfies the required property to showing that $a_n(t) = A(2n/tk(n))$ satisfies it; this latter task can be performed using the Potter bounds. Thus \eqref{eq:moment_bound} holds.
\end{proof}

\begin{proof}[Proof of Theorem \ref{theo:process}]
It is well known \citep[e.g.,][]{loeve1948fonctions} that, in order for a continuous-time stochastic process to have an almost surely $\cc^\infty$ modification, it is sufficient for the covariance function $C(\tau_1, \tau_2)$ to be infinitely differentiable along the diagonal $s_1 = s_2$. We thus immediately get the desired smoothness result, since $\Cov[R(\tau_1), R(\tau_2)] = \frac{2\gamma^2}{e^{\tau_1} + e^{\tau_2}}$ is smooth on $\RR^2$. The same result tells us that, for any $l, l' \in \NN$,
\begin{equation}
\label{eq:deriv}
\Cov\left[R^{(l)}(\tau), R^{(l')}(\tau)\right] = \frac{\partial^{l+l'}}{\partial u^{l}\partial v^{l'}} \Cov\left[R(u), R(v)\right] \bigg|_{u = v = \tau},
\end{equation}
which gives us the stated covariance result. The joint normality of $R$ and $R'$ and the expectation result follow directly from \eqref{eq:process}.
\end{proof}

\begin{proof}[Proof of Theorem \ref{theo:thresh}]
Recall that
$$ {\htau} = \argmin_\tau Z_\tau, \text{ where } Z_\tau := \left(R_\tau'\right)^2 + \frac{\gamma^2}{2e^\tau}; $$
we need to verify that $\htau$ is almost surely finite. To do so, it suffices to show that
\begin{equation}
\label{eq:goal}
\lim_{\tau \rightarrow \pm \infty} Z_\tau \stackrel{\text{a.s.}}{=} +\infty.
\end{equation}
Now, $Z_\tau$ goes deterministically to infinity as $\tau \rightarrow - \infty$, so we only need to check the $\tau \rightarrow + \infty$ limit. In order to verify that $Z_\tau$ goes to infinity it suffices to verify that $R_\tau'$ does. 

By using the same argument as in the proof of Theorem \ref{theo:process}, we see that $R'_\tau$ is a Gaussian process with moments
$$ \EE\left[R_\tau'\right] = -\rho \, 2^\rho\lambda\Gamma(1 - \rho) e^{-\rho\tau} \text{ and } \Cov\left[R_\sigma', \, R_\tau'\right] = 4 \gamma^2\frac{e^{\sigma + \tau}}{\left(e^{\sigma} + e^{\tau}\right)^3}. $$
Because $\rho < 0$, we see that $\lim_{\tau \rightarrow +\infty} \EE\left[R_\tau'\right] = +\infty$. It is easy to verify that $R_\tau' - \EE\left[R_\tau'\right]$ must visit the $[-1, 1]$ interval infinitely many times as $\tau \rightarrow +\infty$. Thus, by continuity of $R_\tau'$, if we show that
\begin{equation}
\label{eq:cross}
\lim_{T \rightarrow \infty} \PP\Big[ R_\tau' - \EE\left[R_\tau'\right]  = -1 \text{ for some } \tau \geq T\Big] = 0,
\end{equation}
we can conclude that \eqref{eq:goal} holds.

For convenience, let $\tRpt = R'_\tau - \EE\left[R'_\tau\right]$, and let $N_T$ be the number of times $\tRpt$ crosses -1 for $\tau \geq T$. From Rice's formula \citep[see][]{adler2007random}, we know that
\begin{align*}
\EE\left[N_T\right]
&= \int_T^\infty \EE\left[\left|\tRpt' \right| \Big | \tRpt = -1\right] f_{\tRpt}(-1) \ d\tau\\
&\leq \int_T^\infty \sqrt{\EE\left[\left(\tRpt' \right)^2 \Big | \tRpt = -1\right]} f_{\tRpt}(-1) \ d\tau,
\end{align*}
where $f_{\tRpt}$ is the marginal density of $\tRpt$ and the second line is an application of Jensen's inequality. The relation \eqref{eq:deriv} implies that
$$ \begin{pmatrix} \tRpt \\ \tRpt' \end{pmatrix} \sim \nn \left(0, \, \frac{\gamma^2}{2e^\tau} \, \begin{pmatrix} 1 & -1/2 \\ -1/2 & 1\end{pmatrix}\right), \text{ and so }
\mathcal{L}\left(\tRpt' \Big| \tRpt = -1\right) = \nn\left(\frac{1}{2}, \, \frac{3}{8}\frac{\gamma^2}{e^\tau}\right).$$
Thus,
$$\EE\left[\left(\tRpt' \right)^2 \bigg | \tRpt = -1\right] = \frac{1}{4} + \frac{3}{8}\frac{\gamma^2}{e^\tau} \leq 1 $$
for all large enough values of $\tau$, meaning that for large enough $T$
$$ \EE\left[N_T\right] \leq \int_T^\infty \frac{\sqrt{2} \, e^{\tau/2}}{\gamma} \, \varphi\left(-\frac{\sqrt{2} \, e^{\tau/2}}{\gamma}\right) \ d\tau, $$
where $\varphi$ is the standard normal density. This integral converges, and so by Markov's inequality \eqref{eq:cross} holds.
\end{proof}

\subsection{Uniform Integrability}

We end with a technical lemma that we have used repeatedly to make uniform integrability arguments.

\begin{lemm}
\label{lemm:small}

Let $X_{1, s} \leq ... \leq X_{s, s}$ be independent order statistics drawn from a distribution $F$ of strictly positive support with extreme value index $\gamma > 0$. Then, for any fixed $k$ and finite $C$,
$$ \lim_{s \rightarrow \infty} \PP\left[X_{s-k, s} < C\right] = 0. $$
Moreover, if there is a constant $\beta > 0$ such that
$$ \lim_{x \rightarrow 0} F(x) \cdot x^{-\frac{1}{\beta}} = 0, $$
then, for any $\nu > 0$,
\begin{equation}
\label{eq:small}
\lim_{s \rightarrow \infty} \EE\left[\big|\log[X_{s-k, s}] \big|^\nu\cdot 1\{X_{s-k, s} < C\}\right] = 0.
\end{equation}

\proof
As in the proof of Lemma \ref{lemm:hill}, we write $X_k \eqd U(Y_k)$. Because $U(t) \rightarrow \infty$, the first statement follows directly by applying the strong law of large numbers to $1\{Y_k > r\}$ for a properly chosen $r > 0$. To prove the second part, we see that
$$\EE\left[\log[X_{s-k, s}]_+^\nu \big| X_{s-k, s} < C\right] \leq \log[C]_+^\nu$$
is uniformly bounded, and we already know that $\PP\left[X_{s-k, s} < C\right]$ converges to zero. The hard part of establishing \eqref{eq:small} is thus to establish a uniform bound for $\EE\left[\log[X_{s-k, s}]_-^\nu \big| X_{s-k, s} < C\right]$.

Now, because $\lim_{x \rightarrow 0} F(x) = 0$,
\begin{align*}
\lim_{x \rightarrow 0} F(x) \cdot x^{-\frac{1}{\beta}} = 0
&\iff \lim_{x \rightarrow 0} \left(\frac{F(x)}{1 - F(x)}\right)^\beta \cdot x^{-1} = 0 \\
&\iff \lim_{y \rightarrow 0} y^\beta \cdot U(1 + y)^{-1} = 0,
\end{align*}
where we obtain the last equivalence by writing $y = \frac{F(x)}{1 - F(x)}.$

Without loss of generality, we picked $C$ such that $C = U(t_0)$ for some $t_0$. Because $U(t)$ is monotone increasing, we can find a constant $L$ such that
$$ \left(t - 1\right)^\beta \cdot U(t)^{-1} \leq L $$
for all $1 < t \leq t_0$. (Without loss of generality, let $L = 1$.) This implies that
\begin{align*}
\EE\left[\log[X_{s-k, s}]_-^\nu \big| X_{s-k, s} < C\right]
&\eqd \EE\left[\log[U(Y_{s-k, s})]_-^\nu \big| Y_{s-k, s} < t_0\right] \\
&\leq \beta^\nu \cdot \EE\left[\log\left[Y_{s-k, s} - 1\right]_-^\nu \big| Y_{s-k, s} < t_0\right].
\end{align*}
Now, for any $\alpha \in \RR$, 
$$ \PP[Y_{s-k, s} < \alpha] \leq \binom{s}{s-k} \PP[Y_1 < \alpha]^{s-k}, $$
and so given a fixed $k$ we can pick $s_k$ such that
$$ \PP[Y_{s-k, s} < \alpha] < \PP[Y_1 < \alpha] \text{ for all } 1 \leq \alpha \leq 2 \text{ and } s \geq s_k.$$
Thus, in order to prove our result, it suffices to show that
$$ \EE\left[\log[Y_1 - 1]_-^\nu \big| Y_1 < t_0\right] < \infty.$$
This last expression is just
$$ \int_{1}^{\min\{2, t_0\}} \frac{(-\log (y - 1))^\nu}{y^2} \ dy \ \bigg / \int_{1}^{\min\{2, t_0\}} \frac{dy}{y^2}, $$
which can be shown by calculus to be finite for any $\nu > 0$.
\endproof

\end{lemm}
\end{document}